\documentclass[journal]{IEEEtran}

\usepackage{cite}
\usepackage{empheq}
\usepackage{amsmath}
\usepackage{amsthm}
\usepackage{amssymb}
\usepackage{mathrsfs}
\usepackage{ulem}
\usepackage{graphicx}
\graphicspath{{Figs/}{../pdf/}{../jpeg/}}
\DeclareGraphicsExtensions{.pdf,.jpeg,.png}

\usepackage{cancel}
\usepackage{makecell}
\usepackage{bbding}
\usepackage{booktabs}
\usepackage{array}
\usepackage{multirow}
\usepackage{algorithm}

\newtheorem{theorem}{Theorem}[section]
\newtheorem{lemma}[theorem]{Lemma}
\newtheorem{definition}[theorem]{Definition}
\newtheorem{corollary}[theorem]{Corollary}
\newtheorem{proposition}[theorem]{Proposition}

\newtheorem{assumption}[]{Assumption}

\usepackage{xcolor}

\begin{document}

\title{Robust and Kernelized Data-Enabled Predictive Control for Nonlinear Systems}
\author{\vspace{3mm} Linbin Huang, John Lygeros, and Florian D{\"o}rfler
\thanks{The authors are with the Department of Information Technology and Electrical
Engineering at ETH Z{\"u}rich, Switzerland. (Emails: \text\{\text{linhuang}, \text{jlygeros}, \text{dorfler}\}@ethz.ch)}
\thanks{This research was supported by the SNSF under NCCR Automation and ETH Zurich Funds.}
}

    %\onecolumn
    %\doublespacing
	
	\maketitle
	
	\begin{abstract}
	This paper presents a robust and kernelized data-enabled predictive control (RoKDeePC) algorithm to perform model-free optimal control for nonlinear systems using only input and output data. The algorithm combines robust predictive control and a non-parametric representation of nonlinear systems enabled by regularized kernel methods.
%The RoKDeePC algorithm requires only input/ouput data from the system to calculate the optimal control sequence.
The latter is based on implicitly learning the nonlinear behavior of the system via the representer theorem. Instead of seeking a model and then performing control design, our method goes directly from data to control. This allows us to robustify the control inputs against the uncertainties in data by considering a min-max optimization problem to calculate the optimal control sequence. We show that by incorporating a proper uncertainty set, this min-max problem can be reformulated as a nonconvex but structured minimization problem. By exploiting its structure, we present a projected gradient descent algorithm to effectively solve this problem. Finally, we test the RoKDeePC on two nonlinear example systems -- one academic case study and a grid-forming converter feeding a nonlinear load -- and compare it with some existing nonlinear data-driven predictive control methods.
	\end{abstract}
	
	\begin{IEEEkeywords}
	Data-driven control, kernel methods, nonlinear control, predictive control, robust optimization.
	\end{IEEEkeywords}
	
	\section{Introduction}\label{sec:intro}
	
	Data-driven control has attracted extensive attention in recent years, as it enables optimal control in scenarios, where data is readily available, but the system models are too complex to obtain or maintain; this is often the case, for example, in large-scale energy systems or robotics~\cite{markovsky2021behavioral, huang2019decentralized, elokda2019data}.

One standard approach to data-driven control is indirect. Here one first uses data to identify a system model and then performs control design based on the identified model. This paradigm has a long history in control theory, leading to the developments of system identification (SysID)~\cite{ljung1999system}, model predictive control (MPC)~\cite{morari1999model}, etc. Prediction error, maximum likelihood, and subspace methods are popular SysID approaches, and have been successfully combined with model-based control design in many industrial applications.
Broadly recognized challenges of the indirect approach are incompatible uncertainty estimates and customizing the SysID (e.g., in terms of objective and model order selection) for the ultimate control objective. We refer to~\cite{hjalmarsson1996model, formentin2018core, dorfler2021bridging, campestrini2017data} for different approaches. %Another line of work is to (approximately) learn the nonlinear behavior of the system, i.e., a basis for the set of (finite-time) trajectories.
Moreover, recent advances in systems theory, statistics, machine learning, and deep learning have introduced many promising techniques to learn the behavior of general dynamical systems from data, e.g., by using Koopman operators, Gaussian processes, kernel methods, and neural networks~\cite{williams2015data, rasmussen2003gaussian, pillonetto2014kernel, pascanu2013construct,lian2021koopman,van2021kernel}. The learned models can also be combined with model-based control design (e.g., MPC) to perform optimal control~\cite{lenz2015deepmpc, korda2018linear, kocijan2004gaussian, chen2018optimal, hewing2019cautious, maddalena2021kpc}. However, these methods still belong to indirect data-driven control, and one may not be able to provide deterministic guarantees on the control performance, especially when complicated structures, e.g., neural networks, are used to learn the system's behaviors. We note that compared to using neural networks, function estimation using regularized kernel methods has a tractable formulation and solution, thanks to the well-posedness of the function classes in reproducing kernel Hilbert spaces (RKHSs)~\cite{pillonetto2014kernel, scholkopf2001generalized}.
%However, one key challenge is how to make the step of SysID (learning) compatible with the step of model-based control design. There are research works suggesting that the step of SysID should be customized for the step of control design in order to achieve better performance~\cite{hjalmarsson1996model, formentin2018core}.

%Usually, one needs to robustify the SysID step against the noise in data, for instance, by introducing regularization to avoid over-fitting. The regularization may lead to a biased model. One may assume certainty equivalence and directly use the biased model for control design, but in many cases, the control performance cannot be guaranteed.
%As a remedy, the step of model-based control design can also be robustified to ensure the control performance. However, it may not be straightforward to coordinate these two robustification processes or to quantify the uncertainty propagation.

An alternative formulation is direct data-driven control that aims to go directly from data to control, bypassing the identification step. In recent years, a result formulated by Willems and his co-authors in the context of behavioral system theory, known as the \textit{Fundamental Lemma}~\cite{willems2005note}, has been widely used in formulating direct data-driven control methods such as~\cite{coulson2019data, de2019formulas, berberich2019data, markovsky2021behavioral}. The Fundamental Lemma shows that the subspace of input/output trajectories of a linear time-invariant (LTI) system can be obtained from the column span of a data Hankel matrix, which acts as a data-centric representation of the system dynamics. Unlike indirect data-driven control methods which usually involve a bi-level optimization problem (an inner problem for the SysID step and an outer problem for the control design step), these direct data-driven control methods formulate one single data-to-control optimization problem. In this paradigm, one can conveniently provide robust control against uncertainties in data via regularization and relaxation~\cite{dorfler2021bridging}. For instance, in the data-enabled predictive control (DeePC) algorithm~\cite{coulson2019data}, a proper regularization leads to end-to-end distributional robustness against uncertainties in data~\cite{coulson2019regularized, huang2020quad}. Moreover, performance guarantees on the realized input/output cost can be provided given that the inherent uncertainty set is large enough~\cite{huang2021robust}.
Motivated in part by the desire to extend this direct data-driven approach beyond LTI systems, Fundamental Lemma has been extended for certain classes of nonlinear systems, e.g., Hammerstein-Wiener systems~\cite{berberich2020trajectory}, second order discrete Volterra systems~\cite{rueda2020data}, flat nonlinear systems~\cite{alsalti2021data}, and polynomial time-invariant systems~\cite{markovsky2021data}, which enables the control design of these classes of nonlinear systems. However, to the best of our knowledge, there has not been a version of Fundamental Lemma for general nonlinear systems.

In this paper, we develop a direct data-driven method for nonlinear systems, referred to as \textit{robust and kernelized DeePC} (RoKDeePC), which combines kernel methods with a robustified DeePC framework.
We use regularized kernel methods~\cite{scholkopf2001generalized} to implicitly learn the future behavior of nonlinear systems, then directly formulate a data-to-control min-max optimization problem to obtain robust and optimal control sequences. We prove that the realized cost of the system is upper bounded by the optimization cost of this optimization problem, leading to deterministic performance guarantees.
Furthermore, we demonstrate how this min-max problem can be reformulated as a nonconvex yet structured minimization problem when appropriate uncertainty sets are considered. To enable a real-time implementation, we develop a projected gradient descent algorithm exploiting the problem structure to effectively solve the RoKDeePC problem. Unlike learning-based control using neural networks, our method does not require a time-consuming off-line learning process and has real-time plug-and-play ability to handle nonlinear systems. We test RoKDeePC on two nonlinear example systems -- one academic case study and a grid-forming converter feeding a nonlinear load -- and compare its performance with the Koopman-based MPC method in~\cite{korda2018linear} and MPC that uses a kernel-based predictor as the system model. These methods are all based on linear predictors in the lifted high-dimensional space and thus fair comparisons can possibly be made.

The rest of this paper is organized as follows: in Section~II we review the linear predictor and then introduce the kernel-based nonlinear predictor using input/output data. Section~III proposes the RoKDeePC algorithm and demonstrates its performance guarantees. Section~IV develops gradient descent algorithms to efficiently solve the RoKDeePC. Section~V tests the RoKDeePC algorithm on two example nonlinear systems. We conclude the paper in Section~VI.

	%\textbf{Notation.}
	\section*{Notation}
Let $\mathbb N$ denote the set of positive integers, $\mathbb S^{n}$ $(\mathbb S^{n}_{>0})$ the set of real $n$-by-$n$ symmetric (positive definite) matrices. We denote the index set with cardinality $n\in \mathbb N$ as $[n]$, that is, $[n] = \{ 1, ..., n \}$.
We use $\|x\|$ ($\|A\|$) to denote the (induced) 2-norm of the vector $x$ (the matrix $A$); we use $\|x\|_1$ to denote the 1-norm of the vector $x$; for a vector $x$, we use ${\left\| x \right\|_A^2}$ to denote  $x^\top A x$; for a matrix $A$, we use $\|A\|_Q$ to denote $\|Q^{\frac{1}{2}}A\|$; we use $\|A\|_F$ to denote the Frobenius norm of the matrix $A$. We use $x_{[i]}$ to denote the $i{\rm th}$ element of $x$ and $x_{[i:j]}$ to denote the vector $[x_{[i]}\;x_{[i+1]}\; \cdots \; x_{[j]}]^\top$; we use $A_{[:i]}$ to denote the $i{\rm th}$ column of $A$, $A_{[i:]}$ to denote the $i{\rm th}$ row of $A$, and $A_{[ij]}$ to denote $i{\rm th}$-row $j{\rm th}$-column element of $A$.
We use $\boldsymbol 1_n$ to denote a vector of ones of length $n$, $\boldsymbol 1_{m\times n}$ to denote a $m$-by-$n$ matrix of ones, and $I_n$ to denote an $n$-by-$n$ identity matrix (abbreviated as $I$ when the dimensions can be inferred from the context).
%We use $A^+$ to denote the right inverse of $A$, and $A^\bot = I - A^+A$ to denote the orthogonal projector onto the kernel of $A$.
We use ${\rm col}(Z_0,Z_1,...,Z_\ell)$ to denote the matrix $[Z_0^{\top}\; Z_1^{\top}\; \cdots \; Z_\ell^{\top}]^{\top}$. %We use $\otimes$ to denote the Kronecker product.

	\section{Kernel-Based Nonlinear Predictors}
	
	\subsection{Explicit Predictors for LTI Behavior}
	
	Consider a discrete-time linear time-invariant (LTI) system
	\begin{equation}
		\left\{ \begin{array}{l}
			{x_{t + 1}} = A{x_t} + B{u_t}\\
			{y_t} = C{x_t} + D{u_t},
		\end{array} \right.\,		\label{eq:ABCD}
	\end{equation}
	where $A \in \mathbb{R}^{n \times n}$, $B \in \mathbb{R}^{n \times m}$, $C \in \mathbb{R}^{p \times n}$, $D \in \mathbb{R}^{p \times m}$, $x_t \in \mathbb{R}^n$ is the state of the system at~$t \in \mathbb{Z}_{ \ge 0}$,~$u_{t} \in \mathbb{R}^m$ is the input vector, and $y_{t} \in \mathbb{R}^p$ is the output vector. Recall the extended observability matrix
	\vspace{-0.5mm}
	\begin{equation*}
		\mathscr{O}_{\ell}(A,C) := {\rm col}(C,CA,...,CA^{\ell-1}) \,.
	\end{equation*}
	The \textit{lag} of the system (\ref{eq:ABCD}) is defined by the smallest integer $\ell \in \mathbb{Z}_{ \ge 0}$ such that the observability matrix $\mathscr{O}_{\ell}(A,C)$ has rank $n$, i.e., the state can be reconstructed from $\ell$ measurements.
	%Note that in addition to Hankel matrices, (Chinese) Page matrices can also serve as data-driven predictors \cite{damen1982approximate,huang2019decentralized}.
	%In a data-driven setting, $\ell$ and $n$ are generally unknown, but upper bounds can usually be inferred from knowledge of the system.
	Consider $L,T \in \mathbb{Z}_{ \ge 0}$ with $T \geq L > \ell$ and length-$T$ input and output trajectories of \eqref{eq:ABCD}: $u = {\rm col} (u_0,u_1,\dots u_{T-1})\in \mathbb{R}^{mT}$ and $y ={\rm col} (y_0,y_1, \dots y_{T-1})\in \mathbb{R}^{pT}$. For the inputs $u$, define the Hankel matrix of depth $L$ as
	\begin{equation}
		\mathscr{H}_L(u) := \left[ {\begin{array}{*{20}{c}}
				{{u_0}}&{{u_1}}& \cdots &{{u_{T - L}}}\\
				{{u_1}}&{{u_2}}& \cdots &{{u_{T - L + 1}}}\\
				\vdots & \vdots & \ddots & \vdots \\
				{{u_{L-1}}}&{{u_{L}}}& \cdots &{{u_{T-1}}}
		\end{array}} \right] \,.		
		\label{eq:Hankel_L}
	\end{equation}
	Accordingly, for the outputs define the Hankel matrix $\mathscr{H}_L(y)$. Consider the stacked matrix
	$\mathscr{H}_L(u,y) = \left[\begin{smallmatrix}\mathscr{H}_L(u)\\\mathscr{H}_L(y)\end{smallmatrix}\right]$. By \cite[Corollary 19]{markovskyidentifiability},  the restricted behavior (i.e., the input/output trajectory of length $L$) equals the image of $\mathscr{H}_L(u,y)$ if and only if rank$\left(\mathscr{H}_L(u,y)\right)=mL + n$. %In words, the Hankel matrix $\mathscr{H}_L(u,y)$ composed of a single length-$T$  trajectory parametrizes all length-$L$ trajectories provided that (and only that) rank$\left(\mathscr{H}_L(u,y)\right)=mL + n$.
	%Note that this result extends and includes the original  Fundamental Lemma  \cite[Theorem 1]{willems2005note} which requires controllability and persistency of excitation of order $L+n$ (i.e., $\mathscr{H}_{L+n}(u)$ must have full row rank) as sufficient conditions.
	
	This behavioral result can be leveraged for data-driven prediction and estimation as follows. Consider $T_{\rm ini},N,T \in \mathbb{Z}_{ \ge 0}$, as well as an input/output time series ${\rm col}(u^{\rm{d}},y^{\rm{d}}) \in \mathbb{R}^{(m+p)T}$ so that rank$\left(\mathscr{H}_{T_{\rm ini}+N}(u^{\rm{d}},y^{\rm{d}})\right)=m(T_{\rm ini}+N) + n$. Here the superscript ``d'' denotes data collected offline, and the rank condition is met by choosing $u^{\rm{d}}$ to be persistently exciting of sufficiently high order and by assuming that the system is controllable and observable. The Hankel matrix $\mathscr{H}_{T_{\rm ini}+N}(u^{\rm{d}},y^{\rm{d}})$ can be partitioned into
	\begin{equation*}
		\left[ {\begin{array}{*{20}{c}}
				{{U_{\rm P}}}\\
				{{U_{\rm F}}}
		\end{array}} \right] := \mathscr{H_c}_{T_{\rm ini}+N}(u^{\rm{d}}) \quad \text{and} \quad \left[ {\begin{array}{*{20}{c}}
				{{\bar Y_{\rm P}}}\\
				{{\bar Y_{\rm F}}}
		\end{array}} \right] := \mathscr{H_c}_{T_{\rm ini}+N}(y^{\rm{d}})\,,		\label{eq:partition_Huy}
	\end{equation*}
	where $U_{\rm P} \in \mathbb{R}^{mT_{\rm ini} \times H_c}$, $U_{\rm F} \in \mathbb{R}^{mN \times H_c}$, $\bar Y_{\rm P} \in \mathbb{R}^{pT_{\rm ini} \times H_c}$, $\bar Y_{\rm F} \in \mathbb{R}^{pN \times H_c}$, and $H_c = T-T_{\rm ini}-N+1$. In the sequel, the data in the partition with subscript P (for ``past'') will be used to implicitly estimate the initial condition of the system, whereas the data with subscript F will be used to predict the ``future'' trajectories. In this case, $T_{\rm ini}$ is the length of an initial trajectory measured in the immediate past during on-line operation, and $N$ is the length of a predicted trajectory starting from the initial trajectory. 
The Fundamental Lemma shows that the image of $\mathscr{H}_{T_{\rm ini}+N}(u^{\rm{d}},y^{\rm{d}})$ spans all length-$(T_{\rm ini}+N)$ trajectories~\cite{willems2005note}, that is,
 ${\rm{col}}(u_{\rm ini},u,\bar y_{\rm ini},y) \in \mathbb R^{(m+p)(T_{\rm ini}+N)}$ is a trajectory of (\ref{eq:ABCD}) if and only if there exists a vector $g \in \mathbb{R}^{H_c}$ so\,that
\begin{equation}\label{eq:Hankel_g}
\left[ {\begin{array}{*{20}{c}}
	{{U_{\rm P}}}\\
	{{\bar Y_{\rm P}}}\\
	{{U_{\rm F}}}\\
	{{\bar Y_{\rm F}}}
	\end{array}} \right]g = \left[ {\begin{array}{*{10}{c}}
	{{u_{\rm ini}}}\\
	{{\bar y_{\rm ini}}}\\
	u\\
	y
	\end{array}} \right]\,.		
\end{equation}
%where $u_{\rm ini} \in \mathbb R^{mT_{\rm ini}}$, $y_{\rm ini} \in \mathbb R^{pT_{\rm ini}}$, $u \in \mathbb R^{mN}$, and $y \in \mathbb R^{pN}$.
The initial trajectory ${\rm{col}}(u_{\rm ini},\bar y_{\rm ini}) \in \mathbb R^{(m+p)T_{\rm ini}}$ can be thought of as setting the initial condition for the future (to be predicted) trajectory ${\rm{col}}(u,y)\in \mathbb R^{(m+p)N}$. In particular, if $T_{\rm ini} \ge \ell$, for every given future input trajectory $u$, the future output trajectory $y$ is uniquely determined through (\ref{eq:Hankel_g}) \cite{markovsky2008data}.

In this case, one can consider an explicit linear predictor as
\begin{equation}\label{eq:linear_predictor}
    y = M{\rm col}(u_{\rm ini},\bar y_{\rm ini},u) \,,
\end{equation}
where $M \in \mathbb{R}^{pN \times [(m+p)T_{\rm ini}+mN]}$ is a linear mapping and can be calculated from data as
\begin{equation}\label{eq:LS_M}
    M = \bar Y_{\rm F}\left[ {\begin{array}{*{20}{c}}
	{{U_{\rm P}}}\\
	{{\bar Y_{\rm P}}}\\
	{{U_{\rm F}}}
	\end{array}} \right]^+ \,,
\end{equation}
where $+$ denotes the pseudoinverse operator.
Note that \eqref{eq:LS_M} is the solution of the following linear regression problem
\begin{equation}
    \begin{array}{cl}
			\displaystyle \mathop {\min}\limits_{M}  &  \left\| \bar Y_{{\rm F}} - M\left[ {\begin{array}{*{20}{c}}
	{{U_{{\rm P}}}}\\
	{{\bar Y_{{\rm P}}}}\\
	{{U_{{\rm F}}}}
	\end{array}} \right] \right\|_F,
		\end{array}
\end{equation}
where ${\rm col}(U_{{\rm P}[:i]},\bar Y_{{\rm P}[:i]},U_{{\rm F}[:i]},\bar Y_{{\rm F}[:i]})$ can be considered as one sample trajectory of ${\rm col}(u_{\rm ini},\bar y_{\rm ini},u,y)$.

In addition to Hankel matrices, one can also use more input/output data to construct (Chinese) Page matrices, mosaic Hankel matrices, or trajectory matrices as data-driven predictors \cite{markovskyidentifiability, coulson2020distributionally, van2020willems}.

\subsection{Explicit Kernel-Based Nonlinear Predictors}

We now consider nonlinear systems, where the future outputs are nonlinear functions of the initial trajectory and the future inputs
\begin{equation}\label{eq:nonlinear_pred}
    y_{[i]} = f_i(u_{\rm ini},\bar y_{\rm ini},u),\; i \in [pN]\,.
\end{equation}
%which also coincides with the nonlinear autoregressive exogenous (NARX) model.
In what follows, we focus on how to find the nonlinear functions $f_i$ that best fit the observed nonlinear trajectories contained in ${\rm col}(U_{{\rm P}},\bar Y_{{\rm P}},U_{{\rm F}},\bar Y_{{\rm F}})$.

To reconcile flexibility of the functions class of $f_i$ and well-poseness of the function estimation problem, we assume that the functions $f_i$ belong to an RKHS~\cite{saitoh2016theory}.

\begin{definition}
    An RKHS over a non-empty set $\mathscr{X}$ is a Hilbert space of functions $f:\mathscr{X} \to \mathbb{R}$ such that for each $x \in \mathscr{X}$, the evaluation functional $E_xf:=f(x)$ is bounded.
\end{definition}

An RKHS is associated with a positive semidefinite kernel, called reproducing kernel, which encodes the properties of the functions in this space.

\begin{definition}
    A symmetric function $K:\mathscr{X} \times \mathscr{X} \to \mathbb{R}$ is called positive semidefinite kernel if, for any $h \in \mathbb{N}$,
    \begin{equation*}
        \sum\limits_{i=1}^{h} \sum\limits_{j=1}^{h} \alpha_i\alpha_jK(x_i,x_j) \ge 0,\; \forall (x_k,\alpha_k) \in (\mathscr{X},\mathbb{R}),\; k \in [h] \,.
    \end{equation*}
    The kernel section of $K$ centered at $x$ is $K_x(\cdot) = K(x,\cdot),\; \forall x \in \mathscr{X}$.
\end{definition}

According to the Moore–Aronszajn theorem~\cite{aronszajn1950theory}, an RKHS is fully characterized by its reproducing kernel. To be specific, if a function $f: \mathscr{X} \to \mathbb{R}$ belongs to an RKHS $\mathcal{H}$ with kernel $K$, then it can be expressed as
\begin{equation}\label{eq:RKHS_f}
    f(x) = \sum\limits_{i=1}^{h}\alpha_iK_{x_i}(x) \,,
\end{equation}
for some $\alpha_i \in \mathbb{R}, \;x_i \in \mathscr{X}$, for all $i \in [h]$, possibly with infinite $h$. Then, the inner product in the RKHS $\mathcal{H}$ reduces to
\begin{equation*}
    \langle f,g \rangle := \sum\limits_{i=1}^{h} \sum\limits_{j=1}^{h'} \alpha_i\beta_jK(x_i,x_j)\,,
\end{equation*}
where $f,g \in \mathcal{H}$ and $g(x) = \sum\nolimits_{j=1}^{h'}\beta_jK_{x_j}(x)$, and the induced norm of $f$ to $\|f\|_{\mathcal{H}}^2 = \sum\nolimits_{i=1}^{h} \sum\nolimits_{j=1}^{h} \alpha_i\alpha_jK(x_i,x_j)$.

We consider now the function estimation problem for the nonlinear system in~\eqref{eq:nonlinear_pred}. Since every column in the data matrix ${\rm col}(U_{{\rm P}},\bar Y_{{\rm P}},U_{{\rm F}},\bar Y_{{\rm F}})$ is a trajectory of the system and thus an input/output sample for~\eqref{eq:nonlinear_pred}, for each $i \in [pN]$, we search for the best (in a regularized least-square sense) $f_i$ in an RKHS $\mathcal{H}$ with kernel $K$
\begin{equation}\label{eq:min_RKHS}
    \begin{array}{cl}
			\displaystyle \mathop {\min}\limits_{f_i \in \mathcal{H}}  &  \sum\limits_{j=1}^{H_c}(\bar Y_{{\rm F}[ij]} - f_i(x_j))^2 + \gamma \|f_i\|_{\mathcal{H}}^2
	\end{array}
\end{equation}
where $x_j = {\rm col}(U_{{\rm P}[:j]},\bar Y_{{\rm P}[:j]},U_{{\rm F}[:j]})$. Note that to simplify the forthcoming developments, the formulation in~\eqref{eq:min_RKHS} does not consider temporal correlation between different indices $i$ and causality (i.e., $y_t$ is not affected by $u_{t+1}$).
One can further consider causality by restricting the arguments of $f_i$ to be $(u_{\rm ini},y_{\rm ini},u_{[1:i]})$, which will lead to more complicated formulations than those in this paper, as different $f_i$ is defined over sets with different dimensions.

The cost function in~\eqref{eq:min_RKHS} includes the prediction errors and a regularization term, which is used to avoid over-fitting and penalize undesired behaviors. We note that since the functions in an RKHS can be parameterized as in~\eqref{eq:RKHS_f}, the representer theorem~\cite{scholkopf2001generalized} ensures that the minimization problem in~\eqref{eq:min_RKHS} can be solved in closed form. In our context, this leads to the following.

\begin{lemma}
    The minimizer of~\eqref{eq:min_RKHS} admits a closed-form prediction model for~\eqref{eq:nonlinear_pred} as
    \begin{equation}\label{eq:K_pred}
        y = \bar Y_{\rm F}(\mathbf{\bar K}+\gamma I)^{-1} k(u_{\rm ini},\bar y_{\rm ini},u) \,,
    \end{equation}
    where $\mathbf{\bar K} \in \mathbb{S}_{\ge 0}^{H_c \times H_c}$ is the Gram matrix such that  $\mathbf{\bar K}_{[ij]} = K(x_i,x_j)$, and $k(\cdot) = {\rm col}(K_{x_1}(\cdot),K_{x_2}(\cdot),\dots,K_{x_{H_c}}(\cdot))$.
\end{lemma}
%\begin{proof}
%The claimed result can be obtained by applying the representer theorem~\cite{scholkopf2001generalized}.
%\end{proof}

Eq.~\eqref{eq:K_pred} serves as a data-driven nonlinear explicit predictor which can be used to predict future behaviors of the system~\eqref{eq:nonlinear_pred} and embedded in predictive control algorithms. We note that the choice of kernel $K$ encodes the properties of the function space. Indeed, \eqref{eq:K_pred} shows that the synthesized functions are linear combinations of the kernel sections centered at the sampling points. Hence, one can choose the kernel $K$ based on a priori knowledge on the system. For example, if one knows that $f_i$ ($i \in [pN]$) are polynomials in inputs and initial data, then a polynomial kernel $K(x_i,x_j) = (x_i^\top x_j + c)^d$ can be used. When the functions are very high-order polynomials or other complex forms, a Gaussian kernel $K(x_i,x_j) = {\rm exp}\left( -\frac{\|x_i-x_j\|^2}{2\sigma^2} \right)$ can be used to span the set of continuous functions.

\section{DeePC with Kernel Methods}
\label{sec:KDeePC}
	
	\subsection{Review of the DeePC algorithm}
	
	The DeePC algorithm proposed in \cite{coulson2019data} directly uses input/output data collected from the unknown LTI system to predict the future behaviour, and performs optimal and safe control without identifying a parametric system representation. More specifically, DeePC solves the following optimization problem to obtain the optimal future control inputs
	\begin{equation} 			\label{eq:DeePC}
		\begin{array}{cl}
			\displaystyle \mathop {{\rm{min}}}\limits_{u \in \mathcal{U},y \in \mathcal{Y} \atop g }  \left\{ \ell(u) + {\left\| {y - r} \right\|_Q^2}  \, \Bigg|  \,\,   {\begin{bmatrix}
	{{U_{\rm P}}}\\
	{{\bar Y_{\rm P}}}\\
	{{U_{\rm F}}}\\
	{{\bar Y_{\rm F}}}
	\end{bmatrix}}g =  {\begin{bmatrix}
	{{u_{\rm ini}}}\\
	{{\bar y_{\rm ini}}}\\
	u\\
	y
	\end{bmatrix}}  \right\},
		\end{array}
	\end{equation}
where $\ell: \mathbb{R}^{mN} \to \mathbb{R}$ is the cost function of control actions, for instance, $\ell(u) = \|u\|_R$ or $\ell(u) = \|u\|_R^2$.
	The sets of input/output constraints are defined as $\mathcal U = \{ u \in \mathbb{R}^{mN} \ | \ W_u u  \le w_u \}$ and $\mathcal Y = \{ y \in \mathbb{R}^{pN} \ | \ W_y y  \le w_y \}$. The positive definite matrix~$R \in \mathbb{S}^{mN \times mN}_{>0}$ and positive semi-definite matrix~$Q \in \mathbb{S}^{pN \times pN}_{\ge 0}$ are the cost matrices. The vector $r \in \mathbb{R}^{pN}$ is a prescribed reference trajectory for the future outputs.
	DeePC involves solving the convex quadratic problem (\ref{eq:DeePC}) in a receding horizon fashion, that is, after calculating the optimal control sequence $u^\star$, we apply $(u_t,...,u_{t+k-1}) = (u_0^{\star},...,u_{k-1}^{\star})$ to the system for $k \le N$ time steps, then, reinitialize the problem (\ref{eq:DeePC}) by updating ${\rm col}(\bar u_{\rm ini},\bar y_{\rm ini})$ to the most recent input and output measurements, and setting $t$ to $t+k$, to calculate the new optimal control for the next $k \le N$ time steps. As in MPC, the control horizon $k$ is a design parameter.

	The standard DeePC algorithm in \eqref{eq:DeePC} assumes perfect output data $(\bar Y_{\rm P}, \bar Y_{\rm F}, \bar y_{\rm ini})$ generated from the unknown system \eqref{eq:ABCD}.
	%to compute the optimal control sequence.
	However, in practice, perfect data is not accessible to the controller due to measurement noise.
	We note that regularized DeePC algorithms can be used to provide robustness against disturbances on the data, which often amount to one-norm or two-norm regularization of $g$ in the cost function of~\eqref{eq:DeePC} \cite{coulson2019regularized,huang2021robust}. In this paper, we assume no process noise, i.e., the input data $(U_{\rm P}, U_{\rm F}, u_{\rm ini})$ is perfectly known.
	
	{\textbf{Perfect Data and Noisy Data.}} Throughout the paper, we use $\bar Y_{\rm P}$, $\bar Y_{\rm F}$, and $\bar y_{\rm ini}$ to denote the perfect (noiseless) output data generated from the system, which accurately captures the system dynamics; we use $\hat Y_{\rm P}$, $\hat Y_{\rm F}$, and $\hat y_{\rm ini}$ to denote the corresponding noisy output data. We use $\mathbf{\bar K}$ to denote the Gram matrix calculated from perfect output data, and $\mathbf{\hat K}$ to denote the Gram matrix calculated from noisy output data.

	\subsection{RoKDeePC algorithm}
	
	The DeePC algorithm above implicitly incorporates a linear predictor. In this section, we develop a RoKDeePC algorithm using an implicit and robustified version of the kernel-based nonlinear predictor~\eqref{eq:K_pred} to perform data-driven optimal control for nonlinear systems. We first consider the following certainty-equivalence MPC problem using the kernel predictor in~\eqref{eq:K_pred} (referred to as kernel-based MPC in this paper)
	\begin{equation}\label{eq:KMPC}
\begin{array}{cl}
\mathop{\min}\limits_{u \in \mathcal{U},y \in \mathcal{Y}}& \ell(u) + \|y-r\|_Q\\
{\rm s.t.}& \eqref{eq:K_pred} \,.
\end{array}
    \end{equation}

    %We consider here weighted 2-norms in the cost function rather than squared 2-norms to simplify our following results.

    In what follows, we assume that when perfect data is available, the predictor~\eqref{eq:K_pred} can faithfully predict the future behavior of the nonlinear dynamics~\eqref{eq:nonlinear_pred}.

    \begin{assumption}\label{as:2}
When perfect data $(\bar Y_{\rm P}, \bar Y_{\rm F}, \bar y_{\rm ini})$ is available in the kernel-based nonlinear predictor~\eqref{eq:K_pred}, the prediction error is bounded, i.e.,
\begin{equation}\label{eq:y_sys}
\| y_{\rm sys} - y_{\rm prd} \|_Q \le \beta_e \,,
\end{equation}
where $y_{\rm prd} = \bar Y_{\rm F}(\mathbf{\bar K}+\gamma I)^{-1} k(u_{\rm ini},\bar y_{\rm ini},u_{\rm sys})$ is the predicted output trajectory based on~\eqref{eq:K_pred}, $u_{\rm sys}$ is a control sequence applied to the system, $y_{\rm sys}$ is the realized output trajectory of the system in response to $u_{\rm sys}$, and $\beta_e$ is the error bound.
\end{assumption}

    Assumption~\ref{as:2} implicitly requires that sufficiently rich data is used to capture the system dynamics \cite{maddalena2021deterministic}. Moreover, one can possibly use more data to obtain a more accurate prediction and thus reduce $\beta_e$. Since perfect data is available, this assumption can generally be satisfied if the unknown dynamics $f_i$ are contained in the RKHS of the chosen kernel $K$, e.g., $f_i$ is polynomial and a sufficiently high-order polynomial kernel is used. Otherwise, this assumption can be met by using kernels with the universal approximation property (e.g., a Gaussian kernel) and by implicitly assuming that $f_i$ is continuous and that ${\rm col}(u_{\rm ini},\bar y_{\rm ini},u_{\rm sys})$ is contained in a compact set~\cite{micchelli2006universal}.

    When noisy data is used, the prediction error may grow unbounded, and the obtained control sequence from~\eqref{eq:KMPC} may not result in a robust closed loop.
    %To this end, we develop a RoKDeePC algorithm that provides robust and optimal control sequence to the system.
    As a first step toward addressing this problem, we rewrite~\eqref{eq:K_pred} implicitly as
    \begin{equation}\label{eq:K_predg}
    \begin{bmatrix}
      \mathbf{\bar K}+\gamma I\\
      \bar Y_{\rm F}
    \end{bmatrix}g = \begin{bmatrix}
                       k(u_{\rm ini},\bar y_{\rm ini},u) \\
                       y
                     \end{bmatrix} \,,
    \end{equation}
	where $g$ is uniquely determined by $(\mathbf{\bar K}+\gamma I)g = k(u_{\rm ini},\bar y_{\rm ini},u)$ since $(\mathbf{\bar K}+\gamma I)$ has full rank, and $y = \bar Y_{\rm F}g$ predicts the future behavior of the system. We note that, given ${\rm col}(u_{\rm ini},\bar y_{\rm ini},u)$, there exists a unique $g \in \mathbb{R}^{H_c}$ that satisfies~\eqref{eq:K_predg}. Unlike linear systems, however, given an arbitrary $g$, there may not exist a vector ${\rm col}(u_{\rm ini},\bar y_{\rm ini},u)$ that satisfies~\eqref{eq:K_predg}; for linear systems such a vector exists, thanks to the Fundamental Lemma. Substituting~\eqref{eq:K_predg} into~\eqref{eq:KMPC} leads to
\begin{equation}\label{eq:KMPC1}
\begin{array}{cl}
\mathop{\min}\limits_{u \in \mathcal{U},y \in \mathcal{Y},g}& \ell(u) + \|y-r\|_Q\\
{\rm s.t.}& \eqref{eq:K_predg} \,.
\end{array}
\end{equation}

The equality constraint~\eqref{eq:K_predg} assumes perfect data to obtain a satisfactory prediction. When perfect data is not available, one may need to relax this equality constraint and ensure feasibility, for instance, by penalizing its violation in the cost function (i.e., soft constraints)
\begin{equation}\label{eq:KMPC2}
\mathop{\min}\limits_{u \in \mathcal{U} \atop g: \bar Y_{\rm F}g \in \mathcal{Y}} \ell(u) + \|\bar Y_{\rm F}g -r\|_Q + \lambda_k \| (\mathbf{\bar K} + \gamma I)g - k(u_{\rm ini},\bar y_{\rm ini},u) \| \,,
\end{equation}
where $\lambda_k$ is a penalty parameter. We note that with a sufficiently large $\lambda_k$, the solution to~\eqref{eq:KMPC2} is equal to the solution to~\eqref{eq:KMPC1}~\cite{kerrigan2000soft}. However, a robust solution may still not be obtained when noisy output data is used.
Hence, we further robustify~\eqref{eq:KMPC2} and propose the following RoKDeePC formulation that leverages noisy output data for nonlinear, robust, and optimal control
%\begin{equation}\label{eq:K_DeePC}
%\begin{array}{cl}
%\mathop{\min}\limits_{u \in \mathcal{U}, g } & \hspace{-2mm} \mathop{\max}\limits_{[\Delta_1\;\Delta_2] \in \mathcal{\hat D}} \;\; \mathop{\min}\limits_{y \in \mathcal{Y},\sigma_k} \;\; \|u\|_R + \sqrt{\|y-r\|_Q^2 + \lambda_k \| \sigma_k \|^2} \vspace{2mm} \\
%{\rm s.t.}& \hspace{-5mm} \footnotesize{\bigg(\begin{bmatrix}
%      \mathbf{\hat K}+\gamma I\\
%      \hat Y_{\rm F}
%    \end{bmatrix}+\underbrace{\begin{bmatrix}\Delta_{11}\\ \Delta_{12}\end{bmatrix}}_{\Delta_1}\bigg)g = \begin{bmatrix}
%                       k(u_{\rm ini},\hat y_{\rm ini},u) + \sigma_k \\
%                       y
%                     \end{bmatrix}+ \underbrace{\begin{bmatrix}\Delta_{21}\\ \Delta_{22}\end{bmatrix}}_{\Delta_2}},
%\end{array}
%\end{equation}
%\begin{equation}\label{eq:K_DeePC}
%\begin{array}{cl}
%\mathop{\min}\limits_{u \in \mathcal{U}, g } & \hspace{-2mm} \mathop{\max}\limits_{[\Delta_K\;\Delta_k] \in \mathcal{D}_1 \atop \Delta_Y \in \mathcal{D}_2} \;\; \mathop{\min}\limits_{y \in \mathcal{Y},\sigma_k} \;\; \ell(u) + {\|y-r\|_Q + \lambda_k \| \sigma_k \|} \vspace{1mm} \\
%{\rm s.t.}& \hspace{-3mm} \begin{bmatrix}
%      \mathbf{\hat K} + \Delta_K+\gamma I\\
%      \hat Y_{\rm F} + \Delta_Y
%    \end{bmatrix}g = \begin{bmatrix}
%                       k(u_{\rm ini},\hat y_{\rm ini},u) + \Delta_k + \sigma_k \\
%                       y
%                     \end{bmatrix} ,
%\end{array}
%\end{equation}
\begin{equation}\label{eq:K_DeePC1}
\begin{array}{cl}
\mathop{\min}\limits_{u \in \mathcal{U}, g \in \mathcal{G} } & \hspace{-2mm} \mathop{\max}\limits_{[\Delta_K\;\Delta_k] \in \mathcal{D}_1 \atop \Delta_Y \in \mathcal{D}_2} \;\; \ell(u) + \|(\hat Y_{\rm F} + \Delta_Y)g - r\|_Q \vspace{1mm} \\
& \hspace{-1mm} + \lambda_k \| (\mathbf{\hat K} + \Delta_K+\gamma I)g - (k(u_{\rm ini},\hat y_{\rm ini},u) + \Delta_k) \| , \vspace{1mm}
\end{array}
\end{equation}
where $[\Delta_K \; \Delta_k]$ and $\Delta_Y$ are disturbances added to the data matrices, which respectively reside in prescribed disturbance sets $\mathcal{D}_1$ and $\mathcal{D}_2$, and $\mathcal{G} = \{ g \; | \; (\hat Y_{\rm F} + \Delta_Y)g \in \mathcal{Y}, \; \forall \Delta_Y \in \mathcal{D}_2 \}$ robustifies the output constraints.
%The optimization problem considers additive disturbances in the equality constraints and thus provides robustness to the noise in data. We consider $y$ and $\sigma_k$ as second-stage decision variables rather than first-stage variables to relax the problem and reduce the conservativeness, similar to related formulations in~\cite{huang2021robust}.

%Notice that one can eliminate the second-stage variables and compactly rewrite \eqref{eq:K_DeePC} as a min-max problem

%\begin{equation}\label{eq:K_DeePC1}
%\begin{array}{cl}
%\mathop{\min}\limits_{u \in \mathcal{U}, g \in \mathcal{G} } & \hspace{-2mm} \mathop{\max}\limits_{[\xi_1\;\xi_2] \in \mathcal{D}} \;\; \| (A+\xi_1)g - (b(u)+\xi_2) \| + \|u\|_R \,,
%\end{array}
%\end{equation}
%where $[\xi_1\;\xi_2] = {\rm diag}(\lambda_k^{\frac{1}{2}}I,Q^{\frac{1}{2}}) [\Delta_1\; \Delta_2]$,
%$\mathcal{G} = \{g\; | \; (\hat Y_{\rm F}+\Delta_{12})g - \Delta_{22} \in \mathcal{Y},\; \forall [{\rm col}(\Delta_{11},\Delta_{12})\;{\rm col}(\Delta_{21},\Delta_{22})] \in \mathcal{\hat D} \}$,
%$\mathcal{D} = \{ [\xi_1\;\xi_2] \; | \; {\rm diag}(\lambda_k^{-\frac{1}{2}}I,Q^{-\frac{1}{2}}) [\xi_1\;\xi_2] \in \mathcal{\hat D} \}$, and
%\begin{equation*}
%A = \begin{bmatrix}
%      \lambda_k^{\frac{1}{2}}(\mathbf{\hat K}+\gamma I)\\
%      Q^{\frac{1}{2}}\hat Y_{\rm F}
%    \end{bmatrix},  \;  b(u) = \begin{bmatrix}
%      \lambda_k^{\frac{1}{2}}k(u_{\rm ini},\hat y_{\rm ini},u)\\
%      Q^{\frac{1}{2}}r
%    \end{bmatrix}.
%\end{equation*}

The above min-max formulation considers the worst-case scenario of the uncertainties added to the data matrices, thereby leading to a robust solution. We will later show that performance guarantees can also be provided by employing this formulation.
Note that min-max optimization problems could be in general difficult to solve, however, \eqref{eq:K_DeePC1} admits a tractable reformulation as a minimization problem if appropriate disturbance sets $\mathcal{D}_1$ and $\mathcal{D}_2$ are considered. The following result shows that the inner max problem of~\eqref{eq:K_DeePC1} can be resolved analytically for appropriate disturbance sets $\mathcal{D}_1$ and $\mathcal{D}_2$, resulting in regularized problem formulations.

\begin{lemma}
Consider $\mathcal{D}_{F1} = \{ [\Delta_K\;\Delta_k] \; | \; \| [\Delta_K\;\Delta_k] \|_F \le \rho_1 \}$, and $\mathcal{D}_{F2} = \{ \Delta_Y \; | \; \| Q^{\frac{1}{2}}\Delta_Y \|_F \le \rho_2 \}$. Given a vector $u \in \mathcal{U}$ and a vector $g \in \mathcal{G}$, it holds that
\begin{equation}\label{eq:max1}
\begin{split}
& \mathop{\max}\limits_{[\Delta_K\;\Delta_k] \in \mathcal{D}_{F1}} \| (\mathbf{\hat K} + \Delta_K+\gamma I)g - (k(u_{\rm ini},\hat y_{\rm ini},u) + \Delta_k) \| \\
& =  \| (\mathbf{\hat K}+\gamma I)g - k(u_{\rm ini},\hat y_{\rm ini},u) \| + \rho_1 \sqrt{\|g\|^2+1} ,
\end{split}
\end{equation}
\begin{equation}\label{eq:max2}
\mathop{\max}\limits_{\Delta_Y \in \mathcal{D}_{F2}} \|(\hat Y_{\rm F} + \Delta_Y)g - r\|_Q = \|\hat Y_{\rm F}g - r\|_Q + \rho_2 \|g\|.
\end{equation}
\end{lemma}
\begin{proof}
%The claimed result can be obtained by applying~\cite[Theorem~3.1]{el1997robust}.
The proof is adapted from~\cite[Theorem~3.1]{el1997robust}. Consider a fixed $g$ and a fixed $u$ in~\eqref{eq:max1}. It follows from triangle inequality that
		\begin{align*}
		   & \mathop{\max}\limits_{[\Delta_K\;\Delta_k] \in \mathcal{D}_{F1}} \| (\mathbf{\hat K} + \Delta_K+\gamma I)g - (k(u_{\rm ini},\hat y_{\rm ini},u) + \Delta_k) \| \\
	 \le  & \| (\mathbf{\hat K} +\gamma I)g - k(u_{\rm ini},\hat y_{\rm ini},u)  \| + \max_{[\Delta_K\;\Delta_k] \in \mathcal{D}_{F1}} \hspace{-1mm} \|  \Delta_K g  - \Delta_k \| \\
	    = &  \| (\mathbf{\hat K} +\gamma I)g - k(u_{\rm ini},\hat y_{\rm ini},u)  \| + \rho_1 \sqrt{ \|g \|^2 + 1}.
		\end{align*}
	Choose $\hat{\Delta} = [\hat{\Delta}_K \ \hat{\Delta}_k] \in \mathcal D_{F1}$ such that
	\[
	[\hat{\Delta}_K \ \hat{\Delta}_k]  = \frac{\rho_1 \omega}{\sqrt{ \|g \|^2 + 1}} [ g^\top \ -1 ], \quad \text{where}
	\]
	\[
	\omega = \begin{cases}
	 \frac{(\mathbf{\hat K} +\gamma I)g - k(u_{\rm ini},\hat y_{\rm ini},u)}{{\|(\mathbf{\hat K} +\gamma I)g - k(u_{\rm ini},\hat y_{\rm ini},u)\|}} &  \hspace{-2mm} \text{if } (\mathbf{\hat K} +\gamma I)g \ne k(u_{\rm ini},\hat y_{\rm ini},u) \\
	 \text{any unit-norm vector} & \hspace{-2mm} \text{otherwise}.
	\end{cases}
	\]
	Since $\hat \Delta$ is a rank-one matrix, we have $\|\hat{\Delta} \|_F = \|\hat{\Delta} \| = \rho_1 $, and
	\begin{align*}
	        & \| (\mathbf{\hat K} + \hat \Delta_K+\gamma I)g - (k(u_{\rm ini},\hat y_{\rm ini},u) + \hat \Delta_k) \| \\
	    = & \| (\mathbf{\hat K} +\gamma I)g - k(u_{\rm ini},\hat y_{\rm ini},u)  \| +  \left\| \hat{\Delta}_K g  - \hat{\Delta}_k \right\| \\
	    = & \| (\mathbf{\hat K} +\gamma I)g - k(u_{\rm ini},\hat y_{\rm ini},u) \| + \rho_1 \sqrt{ \|g \|^2 + 1} \\
		\le & \mathop{\max}\limits_{[\Delta_K\;\Delta_k] \in \mathcal{D}_{F1}} \| (\mathbf{\hat K} + \Delta_K+\gamma I)g - (k(u_{\rm ini},\hat y_{\rm ini},u) + \Delta_k) \|,
	\end{align*}
	where the first equality follows from triangle inequality and the fact that $((\mathbf{\hat K} +\gamma I)g - k(u_{\rm ini},\hat y_{\rm ini},u))$ is a nonnegative scalar multiple of $(\hat{\Delta}_K g  - \hat{\Delta}_k)$. Since the lower bound coincides with the upper bound on the maximization problem~\eqref{eq:max1} for any fixed $g$ and $u$, the equality in~\eqref{eq:max1} holds. By following a similar process, one can prove the equality in~\eqref{eq:max2}.
\end{proof}

\begin{corollary}[A tractable formulation for RoKDeePC]
By considering $\mathcal{D}_1 = \mathcal{D}_{F1}$ and $\mathcal{D}_2 = \mathcal{D}_{F2}$ in~\eqref{eq:K_DeePC1}, it holds that
\begin{equation}\label{eq:K_DeePC2}
\begin{split}
\eqref{eq:K_DeePC1} =  \mathop{\min}\limits_{u \in \mathcal{U}, g \in \mathcal{G} } \; & \ell(u) + \|\hat Y_{\rm F}g - r\|_Q + h(g) \\
& \hspace{-0.7mm} +\lambda_k\| (\mathbf{\hat K}+\gamma I)g - k(u_{\rm ini},\hat y_{\rm ini},u) \|  ,
\end{split}
\end{equation}
where $h(g) = \lambda_k\rho_1 \sqrt{\|g\|^2+1} + \rho_2 \|g\|$ is the regularizer. Moreover, the minimizer of~\eqref{eq:K_DeePC2} coincides with that of~\eqref{eq:K_DeePC1}.
\end{corollary}
%\begin{proof}
%The claimed result can be obtained by substituting \eqref{eq:max1} and \eqref{eq:max2} into \eqref{eq:K_DeePC1}.
%\end{proof}

%\begin{lemma}
%Consider $\mathcal{D}_F = \{ [\xi_1\;\xi_2] \; | \; \| [\xi_1\;\xi_2] \|_F \le \rho \}$. Given a vector $u \in \mathcal{U}$ and a vector $g \in \mathcal{G}$, we have
%\begin{equation}\label{eq:max_eq}
%\mathop{\max}\limits_{[\xi_1\;\xi_2] \in \mathcal{D}_F} \hspace{-2.5mm} \| (A+\xi_1)g - (b(u)+\xi_2) \|
%\hspace{-0.5mm} = \hspace{-0.5mm} \| Ag - b(u) \| + \rho \sqrt{\|g\|^2+1} \,.
%\end{equation}
%\end{lemma}
%\begin{proof}
%The claimed result follows from~~\cite[Theorem~3.1]{el1997robust}.
%\end{proof}
%
%\begin{lemma}
%By considering $\mathcal{D} = \mathcal{D}_F$ in~\eqref{eq:K_DeePC1}, it holds that
%\begin{equation}\label{eq:K_DeePC2}
%\eqref{eq:K_DeePC1} = \mathop{\min}\limits_{u \in \mathcal{U}, g \in \mathcal{G} } \| Ag - b(u) \| + \rho \sqrt{\|g\|^2+1} + \|u\|_R \,.
%\end{equation}
%\end{lemma}
%\begin{proof}
%The claimed result can be obtained by substituting \eqref{eq:max_eq} into \eqref{eq:K_DeePC1}.
%\end{proof}

The above result shows that one can get robust control inputs by simply solving a regularized minimization problem. The proposed RoKDeePC algorithm involves solving~\eqref{eq:K_DeePC2} in a receding horizon manner, similar to the DeePC algorithm. The initial trajectory $(u_{\rm ini},\hat y_{\rm ini})$ should be updated every time before solving~\eqref{eq:K_DeePC2}. We remark that $\mathcal{D}_{F1}$ and $\mathcal{D}_{F2}$ are unstructured uncertainties sets which may lead to conservativeness. For instance, $\mathbf{\hat K}$ is a symmetric matrix, but $\mathcal{D}_{F1}$ does not restrict $\Delta_K$ to be symmetric; $\hat Y_{\rm F}$ may be a Hankel data matrix, but a Hankel structure cannot be imposed on $\Delta_Y$ by considering $\mathcal{D}_{F2}$. As will be discussed in the simulation section, though the considered unstructured sets are conservative, they can generally lead to satisfactory performance.

Notice that the cost function in~\eqref{eq:K_DeePC2} is convex in $g$, but nonconvex in $u$. Hence, in general one may not be able to find the global minimizer. Given a feasible control sequence $u_{\rm sys} \in \mathcal{U}$ (e.g., a local minimizer), we define the minimum cost of~\eqref{eq:K_DeePC2} over $g$ as
\begin{equation}\label{eq:opt_cost_u}
\begin{split}
c_{\rm opt}(u_{\rm sys}) = \mathop{\min}\limits_{g \in \mathcal{G} }\; & \ell(u_{\rm sys}) + \|\hat Y_{\rm F}g - r\|_Q + h(g) \\
& \hspace{-0.7mm} +\lambda_k\| (\mathbf{\hat K}+\gamma I)g - k(u_{\rm ini},\hat y_{\rm ini},u_{\rm sys}) \| .
\end{split}
\end{equation}
%One can use gradient descent methods to possibly obtain a local minimum. We will provide more implementation details in the next sections.

\subsection{Performance guarantees of RoKDeePC}

In what follows, we show that the proposed RoKDeePC algorithm provides deterministic performance guarantees when applying a feasible optimal control sequence to the system, thanks to the min-max formulation in~\eqref{eq:K_DeePC1}. We begin by making the following assumption.

%\begin{assumption}\label{as:1}
%There exists a $[\xi_1\;\xi_2] \in \mathcal{D}$ such that
%\begin{equation}\label{eq:Ab_xi}
%\begin{split}
%\begin{bmatrix}
%      \lambda_k^{\frac{1}{2}}(\mathbf{\hat K}+\gamma I)\\
%      Q^{\frac{1}{2}}\hat Y_{\rm F}
%    \end{bmatrix} + \xi_1 = & \begin{bmatrix}
%      \lambda_k^{\frac{1}{2}}(\mathbf{\bar K}+\gamma I)\\
%      Q^{\frac{1}{2}}\bar Y_{\rm F}
%    \end{bmatrix},  \\
%\begin{bmatrix}
%      \lambda_k^{\frac{1}{2}}k(u_{\rm ini},\hat y_{\rm ini},u)\\
%      Q^{\frac{1}{2}}r
%    \end{bmatrix} + \xi_2 = & \begin{bmatrix}
%      \lambda_k^{\frac{1}{2}}k(u_{\rm ini},\bar y_{\rm ini},u)\\
%      Q^{\frac{1}{2}}r
%    \end{bmatrix}.
%\end{split}
%\end{equation}
%\end{assumption}

\begin{assumption}\label{as:1}
For any ${\rm col}(u_{\rm ini},\bar y_{\rm ini},u_{\rm sys})$ (contained in a compact set), there exists $[\Delta_K\;\Delta_k] \in \mathcal{D}_{F1}$ such that $\mathbf{\hat K} + \Delta_K = \mathbf{\bar K}$ and $k(u_{\rm ini},\hat y_{\rm ini},u_{\rm sys}) + \Delta_k = k(u_{\rm ini},\bar y_{\rm ini},u_{\rm sys})$; there exists $\Delta_Y \in \mathcal{D}_{F2}$ such that $\hat Y_{\rm F} + \Delta_Y = \bar Y_{\rm F}$.

\end{assumption}

% \begin{assumption}\label{as:2}
% The kernel-based nonlinear predictor~\eqref{eq:K_pred} can accurately predict the future behavior of the system when perfect data is available, i.e.,
% \begin{equation}\label{eq:y_sys}
% y_{\rm sys} = \bar Y_{\rm F}(\mathbf{\bar K}+\gamma I)^{-1} k(u_{\rm ini},\bar y_{\rm ini},u_{\rm sys}) \,,
% \end{equation}
% where $u_{\rm sys}$ is the control sequence applied to the system, and $y_{\rm sys}$ is the realized output trajectory.
% \end{assumption}

Assumption~\ref{as:1} indicates that sufficiently large uncertainty sets $\mathcal{D}_{F1}$ and $\mathcal{D}_{F2}$ are used to cover the deviations of the data matrices induced by noise. This is necessary for~\eqref{eq:K_DeePC1} to provide a robust solution against the realized uncertainty in practice.
Since ${\rm col}(u_{\rm ini},\bar y_{\rm ini},u_{\rm sys})$ is contained in a compact set and $k(\cdot)$ is continuous (by choosing a continuous kernel), Assumption~\ref{as:1} can be easily satisfied when the output noise is bounded and sufficiently large disturbance sets $\mathcal{D}_{F1}$ and $\mathcal{D}_{F2}$ are considered, which corresponds to choosing sufficiently large regularization parameters $\rho_1$ and $\rho_2$ in~\eqref{eq:K_DeePC2}.
%We note that Assumption~\ref{as:1} can be easily satisfied by choosing sufficiently large disturbance sets $\mathcal{D}_{F1}$ and $\mathcal{D}_{F2}$, which corresponds to sufficiently large regularization parameters $\rho_1$ and $\rho_2$ in~\eqref{eq:K_DeePC2}.
We define the realized cost of the system as $c_{\rm realized}(u_{\rm sys}) = \ell(u_{\rm sys}) + \|y_{\rm sys}-r\|_Q$ (with $y_{\rm sys}$ defined in Assumption~\ref{as:2}). The following result shows that the realized cost is upper-bounded by the optimization cost in~\eqref{eq:opt_cost_u} plus the prediction error in~\eqref{eq:y_sys}.

\begin{theorem}\label{thm:perf}
If Assumptions~\ref{as:2} and \ref{as:1} hold, then there exists a sufficiently large $\lambda_k$ for~\eqref{eq:opt_cost_u} such that
\begin{equation}\label{eq:perf_ineq}
c_{\rm realized}(u_{\rm sys}) \le c_{\rm opt}(u_{\rm sys}) + \beta_e, \; {\rm for\;any\;feasible\;} u_{\rm sys} \in \mathcal{U} \,,
\end{equation}
where $\beta_e$ is the prediction error bound in Assumption~\ref{as:2}.
\end{theorem}
\begin{proof}
According to the definition of $c_{\rm opt}(u_{\rm sys})$, if Assumption~\ref{as:1} holds, we have
\begin{equation}\label{eq:ineq1}
\begin{split}
c_{\rm opt}(u_{\rm sys}) & =
\mathop{\min}\limits_{g \in \mathcal{G} }  \mathop{\max}\limits_{[\Delta_K\;\Delta_k] \in \mathcal{D}_1 \atop \Delta_Y \in \mathcal{D}_2} \hspace{-1mm}  \ell(u_{\rm sys}) + \|(\hat Y_{\rm F} + \Delta_Y)g - r\|_Q \vspace{1mm} \\
& \hspace{-3mm} + \lambda_k \| (\mathbf{\hat K} + \Delta_K+\gamma I)g - (k(u_{\rm ini},\hat y_{\rm ini},u_{\rm sys}) + \Delta_k) \| \\
& \ge \ell(u_{\rm sys}) + \|\bar Y_{\rm F}g^\star - r\|_Q \\
& \hspace{3.8mm} +\lambda_k\| (\mathbf{\bar K}+\gamma I)g^\star - k(u_{\rm ini},\bar y_{\rm ini},u_{\rm sys}) \| ,
\end{split}
\end{equation}
where $g^\star$ minimizes~\eqref{eq:K_DeePC1} in which $u = u_{\rm sys}$.
Given a vector $u_{\rm sys}$, it follows from~\eqref{eq:K_predg} that there exists a $\bar g$ such that
\begin{equation}\label{eq:eq1}
\begin{bmatrix}
      \mathbf{\bar K}+\gamma I\\
      \bar Y_{\rm F}
    \end{bmatrix}\bar g = \begin{bmatrix}
                       k(u_{\rm ini},\bar y_{\rm ini},u_{\rm sys}) \\
                       y_{\rm prd}
                     \end{bmatrix} ,
\end{equation}
where $y_{\rm prd}$ is the predicted output trajectory from the kernel-based nonlinear predictor.
By defining $\Delta_g = \bar g - g^\star$, we have
\begin{equation}\label{eq:eq2}
\begin{bmatrix}
      \mathbf{\bar K}+\gamma I\\
      \bar Y_{\rm F}
    \end{bmatrix}\Delta_g = \begin{bmatrix}
    \underbrace{k(u_{\rm ini},\bar y_{\rm ini},u_{\rm sys}) - (\mathbf{\bar K}+\gamma I)g^\star}_{\epsilon} \\
                       y_{\rm prd} - \bar Y_{\rm F}g^\star
                     \end{bmatrix} ,
\end{equation}
which implies that
\begin{equation}\label{eq:eq3}
\bar Y_{\rm F} \Delta_g = \bar Y_{\rm F} (\mathbf{\bar K}+\gamma I)^{-1}\epsilon = y_{\rm prd} - \bar Y_{\rm F}g^\star .
\end{equation}
By substituting~\eqref{eq:eq2} and \eqref{eq:eq3} into~\eqref{eq:ineq1} we obtain
\begin{equation}\label{eq:ineq2}
\begin{split}
c_{\rm opt}(u_{\rm sys}) & \ge \ell(u_{\rm sys}) + \|y_{\rm prd} - r - \bar Y_{\rm F}(\mathbf{\bar K}+\gamma I)^{-1}\epsilon\|_Q \\
& \hspace{4mm} + \lambda_k \| \epsilon \| \\
& \ge \ell(u_{\rm sys}) + \|y_{\rm prd} - r \|_Q \\
& \hspace{4mm} - \|\bar Y_{\rm F}(\mathbf{\bar K}+\gamma I)^{-1}\epsilon\|_Q + \lambda_k \| \epsilon \| \\
& \ge \ell(u_{\rm sys}) + \|y_{\rm prd} - r \|_Q \\
& \ge \ell(u_{\rm sys}) + \|y_{\rm sys} - r \|_Q - \|y_{\rm sys} - y_{\rm prd}\|_Q\\ & \ge c_{\rm realized}(u_{\rm sys}) - \beta_e,
\end{split}
\end{equation}
where the second and the forth inequalities hold thanks to the reverse triangle inequality, the third inequality is satisfied if we take $\lambda_k$ large enough (by assumption) to ensure
\begin{equation*}
\begin{split}
& \lambda_k^2 I \succeq (\mathbf{\bar K}+\gamma I)^{-1}\bar Y_{\rm F}^\top Q \bar Y_{\rm F} (\mathbf{\bar K}+\gamma I)^{-1} \,.  \\
%\implies &  \lambda_k \| \epsilon \| \ge \|(\mathbf{\bar K}+\gamma I)^{-1}\epsilon\|_Q ,
\end{split}
\end{equation*}
Hence, $\lambda_k \| \epsilon \| \ge \|(\mathbf{\bar K}+\gamma I)^{-1}\epsilon\|_Q$, and the fifth inequality follows from Assumption~\ref{as:2} and the definition of $c_{\rm realized}(u_{\rm sys})$.
This completes the proof.
\end{proof}

The above result links the optimization cost to the realized cost, namely, the realized cost is always bounded by the optimization cost plus the prediction error even when noisy data is used. This also justifies the design of the cost function of the RoKDeePC in~\eqref{eq:K_DeePC2}, and shows that one should try to find not only a feasible solution but ideally a minimizer (or better the global minimizer) to reduce the optimization cost and possibly the realized cost.

\section{Gradient Descent Methods to Solve RoKDeePC}
\label{sec:GD}

The RoKDeePC problem in~\eqref{eq:K_DeePC2} is nonconvex in $u$ unless the kernel function is linear in $u$. One may use a generic nonlinear optimization toolbox (e.g., IPOPT) to solve this problem. However, solving a nonconvex problem is in general time-consuming, especially considering that the decision variable $g$ in~\eqref{eq:K_DeePC2} can be high-dimensional. We propose customized gradient descent methods exploiting the structure of~\eqref{eq:K_DeePC2} to solve the RoKDeePC problem more efficiently. Although a global minimum may still not be reached, the numerical study in Section~V suggests that these methods generally lead to satisfactory performance with a reasonable running time.

\subsection{Gradient descent algorithm for RoKDeePC}

Notice that the optimization problem in~\eqref{eq:K_DeePC2} is nonconvex in $u$ but convex in $g$. Moreover, the dimension of $g$ is generally much higher than that of $u$. Due to this favorable structure, we apply a gradient descent method to update $u$, while at each iteration, we solve a convex optimization problem to update $g$. Note that the input/output constraints can be handled using projected gradient algorithms.
We summarize the algorithm for solving RoKDeePC below.

\begin{algorithm}[h]
\caption{\normalsize A projected gradient algorithm to solve~\eqref{eq:K_DeePC2}} \label{Algorithm1}

{\bf{Input:}} cost function:
\begin{equation}\label{eq:cost_SOCP}
\begin{split}
c(u,g) = & \; \ell(u) + \|\hat Y_{\rm F}g - r\|_Q + h(g) \\
&+ \lambda_k\| (\mathbf{\hat K}+\gamma I)g - k(u_{\rm ini},\hat y_{\rm ini},u) \|\,,
\end{split}
\end{equation}
%$c(u,g) = \ell(u) + \|\hat Y_{\rm F}g - r\|_Q + h(g) + \lambda_k\| (\mathbf{\hat K}+\gamma I)g - k(u_{\rm ini},\hat y_{\rm ini},u) \|$;
initial vectors $u=u_{(0)}=0$, $g=g_{(0)}=0$; constraint sets $\mathcal{U}$, $\mathcal{G}$; step size $\alpha$; maximal iteration number $i_{\max}$; convergence threshold: $\xi$.

\vspace{2mm}
\hspace{2.5mm} {\bf Initialization:} $i = 0$.

\hspace{2.5mm} {\bf Iterate} until $|c(u_{(i)},g_{(i)}) - c(u_{(i-1)},g_{(i-1)}) | < \xi$ or $i \ge i_{\max}$:
\begin{enumerate}%[ 1)]
	\item $ i \leftarrow i+1 $
    %\item $ u_{(i)} = u_{(i-1)} - \alpha \frac{\partial c(u,g)}{\partial u} \big|_{u = u_{(i-1)}, g = g_{(i-1)}} $ \vspace{2mm} \\
%        if $u_{(i)} \notin \mathcal{U}$: do projection $ u_{(i)} \leftarrow {\rm arg} \mathop {\rm min} \limits_{u \in \mathcal{U}} \|u - u_{(i)}\| $
    \item $ u_{(i)} = {\rm arg} \mathop {\rm min} \limits_{u \in \mathcal{U}} \Big\|u - \Big( u_{(i-1)} - \alpha \frac{\partial c(u,g)}{\partial u} \Big|_{ \begin{subarray}{l}
        u = u_{(i-1)} \\ \vspace{2mm}
                 g = g_{(i-1)}
                 \end{subarray} } \Big) \Big\| $
    \item $g_{(i)} = {\rm arg} \mathop {\rm min} \limits_{g \in \mathcal{G}} c(u_{(i)},g) $
\end{enumerate}
\vspace{2mm}

{\bf{Output:}} (sub)optimal control sequence $u^\star = u_{(i)}$.

\end{algorithm}

In Step~2, the projection of $u$ into the set $\mathcal{U}$ admits a closed-form solution if upper and lower bounds for the elements in $u$ are considered.
Step~3 requires solving a second-order cone program (SOCP), which could make the algorithm time-consuming as one may need to solve the SOCP $i_{\max}$ times.
To reduce the computational burden, in what follows, we consider an equivalent cost function which is quadratic in $g$ and thus admits a closed-form solution for Step~3 when the output constraints are inactive.
%We will show the connection between this new cost function and the original cost function in Algorithm~\ref{Algorithm1}.

\subsection{Quadratic reformulation}

To enable a fast calculation of Step~3, we consider the following cost function that is quadratic in $g$
\begin{equation}\label{eq:quad_cost}
\begin{split}
c_q(u,g) =& \; \ell(u) + \|\hat Y_{\rm F}g - r\|_Q^2 + \lambda_g \|g\|^2 \\
&+ \lambda_k'\| (\mathbf{\hat K}+\gamma I)g - k(u_{\rm ini},\hat y_{\rm ini},u) \|^2 \,. \\
%=& \; \ell(u) + \|Ag-b(u)\|^2 + \lambda_g \|g\|^2 \,,
\end{split}
\end{equation}
%where
%
%$A = \begin{bmatrix} \sqrt{\lambda_k'} (\mathbf{\hat K}+\gamma I)  \\ Q^{\frac{1}{2}} \hat Y_{\rm F} \end{bmatrix}$ and $b(u) = \begin{bmatrix} \sqrt{\lambda_k'} k(u_{\rm ini},\hat y_{\rm ini},u) \\ Q^{\frac{1}{2}} r \end{bmatrix}$.
The following result shows the connection between the cost functions in~\eqref{eq:cost_SOCP} and~\eqref{eq:quad_cost}.

\begin{proposition}\label{Prop1}
For any $u$, if $g^\star \in \mathbb{R}^{H_c}$ is a minimizer of
\begin{equation}\label{eq:min_cq}
  \mathop{\min} \limits_{g \in \mathcal{G}}\; c_q(u,g)\,,
\end{equation}
then $g^\star$ also minimizes
\begin{equation}\label{eq:min_c}
  \mathop{\min} \limits_{g \in \mathcal{G}}\; c(u,g)\,,
\end{equation}
with
\begin{equation}\label{eq:lambda_k_eq}
\hspace{0mm}
\lambda_k \hspace{-0.5mm} = \hspace{-0.5mm} \left\{ \hspace{-2mm}
\begin{array}{ll}
\frac{\lambda_k' \| (\mathbf{\hat K}+\gamma I)g^\star - k(u_{\rm ini},\hat y_{\rm ini},u) \| }{\|\hat Y_{\rm F}g^\star - r\|_Q} & \hspace{-4mm} {\rm if} \; \hat Y_{\rm F}g^\star \ne r \\
\lambda_k' \| (\mathbf{\hat K}+\gamma I)g^\star - k(u_{\rm ini},\hat y_{\rm ini},u) \| & \hspace{-2mm} {\rm otherwise,}
\end{array}
\right.
\end{equation}
\begin{equation}\label{eq:lambda_g_eq}
\frac{\lambda_k \rho_1 \|g^\star\|}{\sqrt{\|g^\star\|^2+1}} + \rho_2 = \left\{ \hspace{-2mm}
\begin{array}{ll}
\frac{\lambda_g \|g^\star\|}{\|\hat Y_{\rm F}g^\star - r\|_Q} & {\rm if} \; \hat Y_{\rm F}g^\star \ne r \\
\lambda_g \|g^\star\| & {\rm otherwise.}
\end{array}
\right.
\end{equation}

\noindent
Moreover, $\lambda_k$ in~\eqref{eq:lambda_k_eq} is monotonically increasing in the $\lambda_k'$ chosen in~\eqref{eq:quad_cost}; if $\rho_2$ ($\rho_1$) remains constant, $\rho_1$ ($\rho_2$) in~\eqref{eq:lambda_g_eq} is monotonically increasing in the $\lambda_g$ chosen in~\eqref{eq:quad_cost}.
\end{proposition}
\begin{proof}
%The claimed result can be obtained by adapting the proof of~\cite[Theorem~1]{huang2020quad}.
See Appendix~\ref{proof_prop1}.
\end{proof}

The above result shows that using the cost function $c_q(u,g)$ in the RoKDeePC is equivalent to setting different parameters ($\lambda_k$, $\rho_1$, and $\rho_2$) to the problem in~\eqref{eq:K_DeePC2} every time solving it during the receding horizon implementation, as the equivalent values of $\lambda_k$, $\rho_1$, and $\rho_2$ are related to $g^\star$. This may lead to conservativeness as we need to choose a robust pair of $(\lambda_g,\lambda_k')$. However, we generally observe excellent performance when using $c_q(u,g)$ as the cost function for RoKDeePC. During the online implementation, after solving~\eqref{eq:min_cq}, one can compute $\lambda_k$, $\rho_1$, and $\rho_2$ (by fixing a desired value of $\rho_1$ or $\rho_2$ and then compute the other) according to~\eqref{eq:lambda_k_eq} and~\eqref{eq:lambda_g_eq}. If the obtained $\lambda_k$, $\rho_1$, and $\rho_2$ are not sufficiently large to satisfy Assumption~\ref{as:1} and Theorem~\ref{thm:perf}, then $\lambda_k'$ and $\lambda_g$ should be increased to implicitly increase $\lambda_k$, $\rho_1$, and $\rho_2$.

Moreover, when the output constraints are inactive, \eqref{eq:min_cq} admits a closed-form solution as
\begin{equation}\label{eq:g_closed_form}
  g^\star = M_r r + M_k k(u_{\rm ini},\hat y_{\rm ini},u) \,,
\end{equation}
where $M_r = [ \hat Y_{\rm F}^\top Q \hat Y_{\rm F} + \lambda_g I + \lambda_k' (\mathbf{\hat K}+\gamma I)^2 ]^{-1} \hat Y_{\rm F}^\top Q$, and $M_k = [ \hat Y_{\rm F}^\top Q \hat Y_{\rm F} + \lambda_g I + \lambda_k' (\mathbf{\hat K}+\gamma I)^2 ]^{-1} \lambda_k' (\mathbf{\hat K}+\gamma I)$. Notice that~\eqref{eq:g_closed_form} requires simply a linear mapping to obtain $g^\star$, which can greatly accelerate the gradient descent algorithm when the output constraints are inactive.
In summary, we consider the following modified RoKDeePC optimization problem
\begin{equation}\label{eq:K_DeePC3}
  \mathop{\min} \limits_{u \in \mathcal{U}, g \in \mathcal{G}}\; c_q(u,g)\,,
\end{equation}
which is related to the original formulation in~\eqref{eq:K_DeePC2} via Proposition~\ref{Prop1}, and we implement the following modifications to Algorithm~\ref{Algorithm1} to solve~\eqref{eq:K_DeePC3}:
\begin{itemize}
  \item Change the cost function from $c(u,g)$ to $c_q(u,g)$;
  \item Change Step~3 in the iteration to: \\
  $g_{(i)} = M_r r + M_k k(u_{\rm ini},\hat y_{\rm ini},u_{(i)})$ \\
  if $g_{(i)} \notin \mathcal{G}$, $g_{(i)} \leftarrow {\rm arg} \mathop {\rm min} \limits_{g \in \mathcal{G}} c_q(u_{(i)},g) $.
\end{itemize}

Step~3 becomes easier to solve if the output constraint set $\mathcal{Y}$ is a nonempty box constraint (i.e., upper and lower bounds for the outputs are considered), as the constraint set $\mathcal{G}$ can be formulated as $\mathcal{G} = \{ g \; | \; G_1g + {\bf 1}\|g\| \le q_1 \}$ for some $G_1$ and $q_1$~\cite{bdv15,huang2021robust}.
The computation cost can be further reduced by assuming that $\|g\|$ is bounded and (conservatively) restrict $\mathcal{G}$ to be a polyhedron $\mathcal{G} = \{ g \; | \; G_1g \le q' \}$ for some $q'$. In our applications, we do not find these assumptions to be limiting.

We remark that one can also assume a kernel function that is nonlinear in $(u_{\rm ini},y_{\rm ini})$ but linear in $u$ such that~\eqref{eq:quad_cost} becomes a convex quadratic program, which can be more efficiently solved using standard solvers. For instance, one may consider
\begin{equation}\label{eq:kernel_linear_u}
\begin{split}
K(x_i,{\rm col}(u_{\rm ini},y_{\rm ini},u)) = &
\exp({-\frac{\|x_{1,i}-{\rm col}(u_{\rm ini},y_{\rm ini})\|^2}{2\sigma^2}}) \\
& + x_{2,i}^\top u \,,
\end{split}
\end{equation}
where $x_i = {\rm col}(x_{1,i},x_{2,i})$. Note that~\eqref{eq:kernel_linear_u} is a combination of two positive semidefinite kernels and thus also satisfies the positive semidefinite property. If the future outputs of the system is linear in $u$, such a kernel can be used to accurately capture the system's dynamics. Otherwise, the obtained kernel-based predictor is the linear-input model that best fits the observed data, which may not lead to an accurate prediction.

\section{Simulation Results}

In this section, we test the RoKDeePC algorithm on two example nonlinear systems to illustrate its effectiveness. We will compare the performance when different kernel functions are used, and we will compare the RoKDeePC algorithm with the traditional DeePC algorithm, the kernel-based MPC in~\eqref{eq:KMPC}, and the Koopman-based MPC in~\cite{korda2018linear}.

\subsection{Example~1: a nonlinear second-order model}
\label{sec:VA}

We start with an academic case study to compare different predictors.
Consider the following discrete-time nonlinear model of a polynomial single-input-single-output system
\begin{equation}\label{eq:testsys1}
  y_t = 4 y_{t-1} u_{t-1} - 0.5 y_{t-1} + 2 u_{t-1} u_t + u_t \,.
\end{equation}
To excite the system and collect input and output data, we inject white noise (mean: 0; variance: 0.01) into the input channel and collect an input/output trajectory of length $T = 600$, which will be used in different predictors below. The observed data may also be subject to additive white measurement noise.

We first test the prediction quality of different predictors, including: i) the linear predictor in~\eqref{eq:linear_predictor}, ii) the kernel predictor in~\eqref{eq:K_pred}, and iii) the Koopman-based predictor in~\cite{korda2018linear}. The parameters in~\eqref{eq:linear_predictor} and~\eqref{eq:K_pred} are: $T_{\rm ini} = 1$, $N = 5$, and $\gamma = 0.01$.
We consider three popular kernels for the kernel predictor: 1) Polynomial kernel $K(x_i,x_j) = (x_i^\top x_j + 1)^{10}$; 2) Gaussian kernel $K(x_i,x_j) = \exp({-\frac{\|x_i-x_j\|^2}{0.4}})$; 3) Exponential kernel $K(x_i,x_j) = \exp({\frac{x_i^\top x_j}{0.2}})$. In this section, the parameters in these kernel functions are chosen to ensure that the kernel predictors have satisfactory performance when noiseless data is available. For instance, the Gaussian kernel has only one parameter to tune, and thus an appropriate value can be easily found with a line search.

For the Koopman-based predictor, one can consider $x_t = (u_{t-1},y_{t-1})$ as the state variables of the above system. Similar to~\cite{korda2018linear}, we consider the lifting functions in the form of
\begin{equation}\label{eq:liftingK}
  \phi_i(x_t,u_t) = \psi_i(x_t) + \mathcal{L}_i(u_t) \,,
\end{equation}
where $\psi_i:\mathbb{R}^n \to \mathbb{R}$ is in general nonlinear but $\mathcal{L}_i: \mathbb{R}^m \to \mathbb{R}$ is linear. The lifting functions $\psi_i$ are chosen to be the linear and quadratic terms of the state variables and 10 thin plate spline radial basis functions (defined by $\psi(x) = \|x - x_0\|^2 {\log}(\|x - x_0\|)$ whose center is $x_0$) with centers selected randomly from $[-1.5,1.5]^2$. Under this setting, the input/output data can be used in the extended dynamic mode decomposition (EDMD)~\cite{williams2015data} to solve for the parametric system representation (Koopman-based predictor). The lifting functions in~\eqref{eq:liftingK} leads to a quadratic program for the MPC formulation, which can be solved efficiently in practice.

%\begin{figure}[!t]
%	\centering
%	\includegraphics[width=3.4in]{Fig_prediction}
%	\vspace{-2mm}
%	%\DeclareGraphicsExtensions.
%	\caption{Prediction comparison under different measurement noise levels. (a) Without measurement noise. (b) White noise (variance: $10^{-3}$).}
%	\vspace{0mm}
%	\label{Fig_prediction}
%\end{figure}

\renewcommand\arraystretch{2}
\begin{table}
	\scriptsize
	\centering
	\caption{Prediction comparison under different noise levels. }
    \vspace{-1mm}
	
    \begin{tabular}{|c|c|c|}	
    \hline
    \makecell{Prediction Error\\ $\sum \|\Delta y\|^2$ } & Without measurement noise & \makecell{With white noise\\ (variance: $10^{-3}$)} \\
    \hline
    Linear predictor & 0.2378 & 0.2384 \\
    \hline
    \makecell{Koopman predictor\\ (linear in $u_t$)} & 0.1927 & 0.2175 \\
    \hline
    \makecell{Kernel predictor \\ (Polynomial kernel)}  & 0.0008 & 0.0486 \\
    \hline
    \makecell{Kernel predictor \\ (Gaussian kernel)}  & 0.0051 & 0.0463 \\
    \hline
    \makecell{Kernel predictor \\ (Exponential kernel)}  & 0.0030 & 0.0456 \\
    \hline
    \end{tabular}
	\vspace{-1mm}
	\label{table:prediction}
\end{table}

\renewcommand\arraystretch{2}
\begin{table*}
	\scriptsize
	\centering
	\caption{Comparison of realized closed-loop performance under different measurement noise levels. }
    \vspace{-1mm}
	
    \begin{tabular}{|c|c|cc|cc|cc|}	
    \hline
    \makecell{noise level} & \makecell{realized\\costs} & \makecell{RoKDeePC\\Gaussian kernel} & \makecell{Kernel-based MPC\\Gaussian kernel} & \makecell{RoKDeePC\\Exponential kernel} & \makecell{Kernel-based MPC\\Exponential kernel} & \makecell{RoKDeePC\\Polynomial kernel} & \makecell{Kernel-based MPC\\Polynomial kernel} \\
    \hline
    \multirow{2}*{\makecell{variance:\\$10^{-4}$}} & \makecell{mean} & 8.92 & 8.51 & 6.07 & 6.16 & 7.31 & 7.36 \\
    \cline{2-8}
    ~ & \makecell{standard\\ deviation} & 4.17 & 3.72 & 1.97 & 2.01 & 3.51 & 3.25 \\
    \hline

    \multirow{2}*{\makecell{variance:\\$10^{-3}$}} & \makecell{mean} & 25.06 & 28.19 & 25.11 & 29.37 & 34.94 & 40.77 \\
    \cline{2-8}
    ~ & \makecell{standard\\ deviation} & 17.99 & 19.24 & 15.31 & 17.66 & 24.97 & 31.43 \\
    \hline

    \multirow{2}*{\makecell{variance:\\$1.5 \times 10^{-3}$}} & \makecell{mean} & 30.76 & \makecell{37.54\\(+22\%)} & 31.89 & \makecell{40.71\\(+28\%)} & 46.76 & \makecell{59.36\\(+27\%)} \\
    \cline{2-8}
    ~ & \makecell{standard\\ deviation} & 24.29 & 28.06 & 22.15 & 27.03 & 36.32 & 50.50 \\
    \hline

    \end{tabular}
	\vspace{-1mm}
	\label{table:cost_compare}
\end{table*}

Table~\ref{table:prediction} compares the open-loop predictions of 50 time steps obtained by employing different predictors. For the linear predictor and the kernel predictor (whose prediction horizon is $N=5$), the time steps from $iN+1$ to $(i+1)N$ ($i = \{1,...,9\}$) are predicted based on the previous predictions. It can be seen that the linear predictor cannot make an accurate prediction due to the strong nonlinearity of the system. The Koopman-based predictor also cannot make an accurate prediction because \eqref{eq:liftingK} assumes independence between $x_t$ and $u_t$, which is not the case in~\eqref{eq:testsys1}. Moreover, we observe similar prediction errors even with more data in the EDMD, e.g., $T=200000$ (data not shown).
The structure in~\eqref{eq:liftingK} enables fast calculations of the associated MPC problem, but restricts the applicability in some classes of nonlinear systems. One can also assume a more general form of $\phi_i$ to improve the prediction accuracy, but it is out of the scope of this paper.

By comparison, all kernel predictors can predict the system's behaviors with satisfactory accuracy. Moreover, when there is no measurement noise, the predictor with polynomial kernel almost perfectly predicts the system's behaviors, as the future outputs are indeed polynomial functions of $(u_{t-1},y_{t-1},u_t)$ due to the bilinear structure in~\eqref{eq:testsys1}.

\begin{figure}[!t]
	\centering
	\includegraphics[width=3.2in]{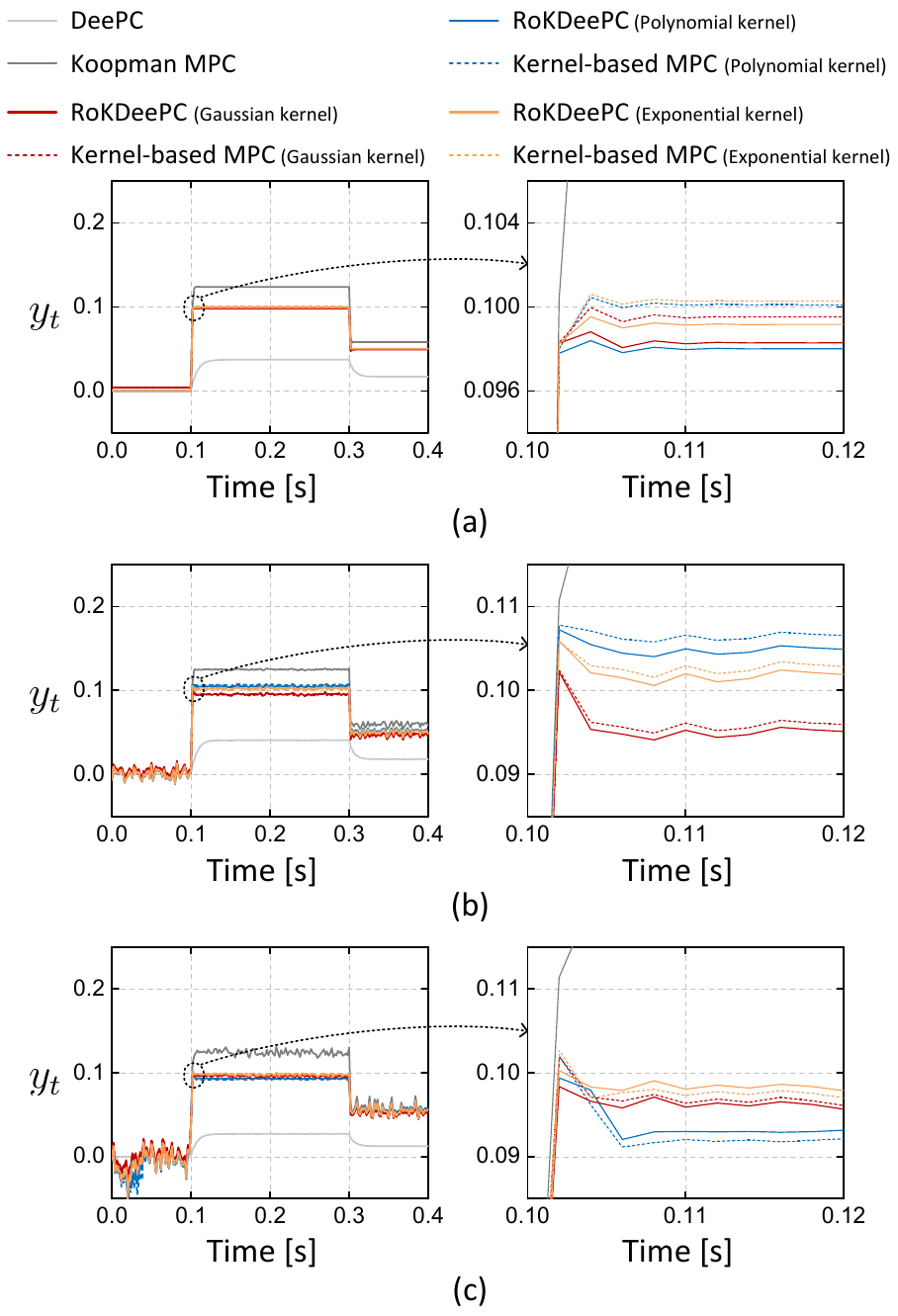}
	\vspace{-2mm}
	%\DeclareGraphicsExtensions.
	\caption{Time-domain responses of the system with different control methods and different measurement noise levels. (a) Without measurement noise. (b) White noise (variance: $10^{-4}$). (c) White noise (variance: $10^{-3}$).}
	\vspace{0mm}
	\label{Fig_timedomain}
\end{figure}

We next compare the time-domain performance when different control methods are applied: 1) DeePC with quadratic regularization~\cite{huang2020quad}; 2) Koopman-based MPC~\cite{korda2018linear}; 3) nominal (i.e., non-robustified) Kernel-based MPC in~\eqref{eq:KMPC} (the cost function is modified to be $\ell(u) + \|y-r\|_Q^2$); 4) RoKDeePC in~\eqref{eq:K_DeePC3}. We consider the following input and output costs in all these methods: $\ell(u) + \|y-r\|_Q^2 = \sum\nolimits_{t=1}^{N} (\|u_t\|^2 + 10^3 \|y_t - r_t\|^2) + \sum\nolimits_{t=2}^{N} 10^2 \|u_t - u_{t-1}\|^2$, where $r_t$ is the output reference over time. This cost function penalizes the control inputs, the tracking errors, and the rates of changes of control inputs. The other parameters of the RoKDeePC are: $\lambda'_k = 10^8$, $\lambda_g = 1$. The control horizon is $k=1$ in all the methods. We assume no input/output constraints and focus on comparing the tracking performance. Hence, no projection is needed in the gradient descent algorithm.

Fig.~\ref{Fig_timedomain} shows the time-domain responses of the system when different control methods are applied, where the output reference $r_t$ steps from $0$ to $0.1$ at $t = 0.1{\rm s}$, and then to $0.05$ at $t = 0.3{\rm s}$. The sampling time of the discrete-time system is 2ms.
It can be seen that when DeePC or Koopman-based MPC is applied, the output cannot accurately track the reference because the linear predictor and the Koopman predictor (linear in $u_t$) cannot make an accurate prediction (Table~\ref{table:prediction}).
%We note that the lifting functions in~\eqref{eq:liftingK} leads to a quadratic program for the MPC formulation, and assuming a more general form of $\phi_i$ will result in a much more complex nonlinear program whose computational tractability within an MPC is questionable and whose implementation is beyond the scope of this paper.
By comparison, the RoKDeePC and the kernel-based MPC both achieve satisfactory performance even under measurement noise.

We remark that with the gradient descent algorithm in Section~\ref{sec:GD}, the RoKDeePC has a reasonable running time. On an Intel Core i7-9750H CPU with 16GB RAM, it requires approximately 0.2s to solve~\eqref{eq:K_DeePC3} (with any of the considered kernels) every time in the above simulations. The kernel-based MPC requires similar time to solve~\eqref{eq:KMPC} using a similar gradient descent algorithm. By comparison, the Koopman-based MPC and DeePC only require to solve quadratic programs, which take about 2ms and 10ms respectively using OSQP~\cite{osqp} as the solver. However, their performance is vastly inferior.

To compare the robustness of the RoKDeePC and the kernel-based MPC, the above simulations were repeated 100 times with different datasets to construct the Hankel matrices and different random seeds to generate the measurement noise. Table~\ref{table:cost_compare} displays the mean and standard deviation of realized closed-loop costs, i.e., $\sum\nolimits_{t=1}^{200} (\|\bar u_t\|^2 + 10^3 \|\bar y_t - r_t\|^2) + \sum\nolimits_{t=2}^{200} 10^2 \|\bar u_t - \bar u_{t-1}\|^2$ where $\bar u_t$ and $\bar y_t$ are measured from the system in the simulations.
It shows that with the increase of the measurement noise level, the RoKDeePC achieves superior performance than the (certainty-equivalence) kernel-based MPC. For instance, when the variance of the measurement noise is $1.5 \times 10^{-3}$, the averaged closed-loop costs of the RoKDeePC are respectively $22\%$ (Gaussian kernel), $28\%$ (Exponential kernel), and $27\%$ (Polynomial kernel) lower than those of the kernel-based MPC. We attribute this performance gap to the robustness of the RoKDeePC, namely, the control sequence obtained from RoKDeePC is robust against the uncertainties in data; by comparison, the kernel-based MPC implicitly assumes certainty equivalence and may not lead to a robust solution. Moreover, we observe that the Gaussian kernel and Exponential kernel have better performance than the Polynomial kernel under noisy measurements, even though the future behaviors of the systems can be described by polynomial functions.

%\begin{figure}[!t]
%	\centering
%	\includegraphics[width=3.0in]{Fig_hist}
%	\vspace{-2mm}
%	%\DeclareGraphicsExtensions.
%	\caption{Comparison of realized closed-loop performance under different measurement noise levels. (a) White noise (variance: $10^{-4}$). (b) White noise (variance: $10^{-3}$). (c) White noise (variance: $1.5 \times 10^{-3}$).}
%	\vspace{0mm}
%	\label{Fig_hist}
%\end{figure}

\subsection{Example~2: A grid-forming converter with nonlinear load}

We now test the performance of the RoKDeePC in a grid-forming converter that is feeding a local nonlinear load, as shown in Fig.~\ref{Fig_converter}. The grid-forming converter employs virtual synchronous machine control for synchronization, which emulates the swing equation of synchronous generators to generate the frequency\cite{dorfler2012synchronization, d2015virtual, yang2020placing}.
The converter is connected to a local nonlinear load, whose active and reactive power consumptions are (in per-unit values)
\begin{equation*}
\begin{split}
P_{\rm Load} = \; & 0.3 + 0.2 U^3 + 10 \Delta U^2 + 5 \Delta U \,, \\
Q_{\rm Load} = \; & 0.04 + 8 \Delta U^2 + 2 \Delta U \,,
\end{split}
\end{equation*}
where $U$ is the terminal voltage magnitude of the load, and $\Delta U = U - U_0$ is the voltage deviation from its nominal value $U_0 = 1{(\rm p.u.)}$. The other system parameters are: $L_F = 0.2{(\rm p.u.)}$, $C_F = 0.07{(\rm p.u.)}$, $L_g = 0.4{(\rm p.u.)}$, $\omega = 100\pi {\rm (rad/s)}$, $P_0 = 0.5{(\rm p.u.)}$, $J = 0.02$, and $D = 0.08$.

\begin{figure}[!t]
	\centering
	\includegraphics[width=3.5in]{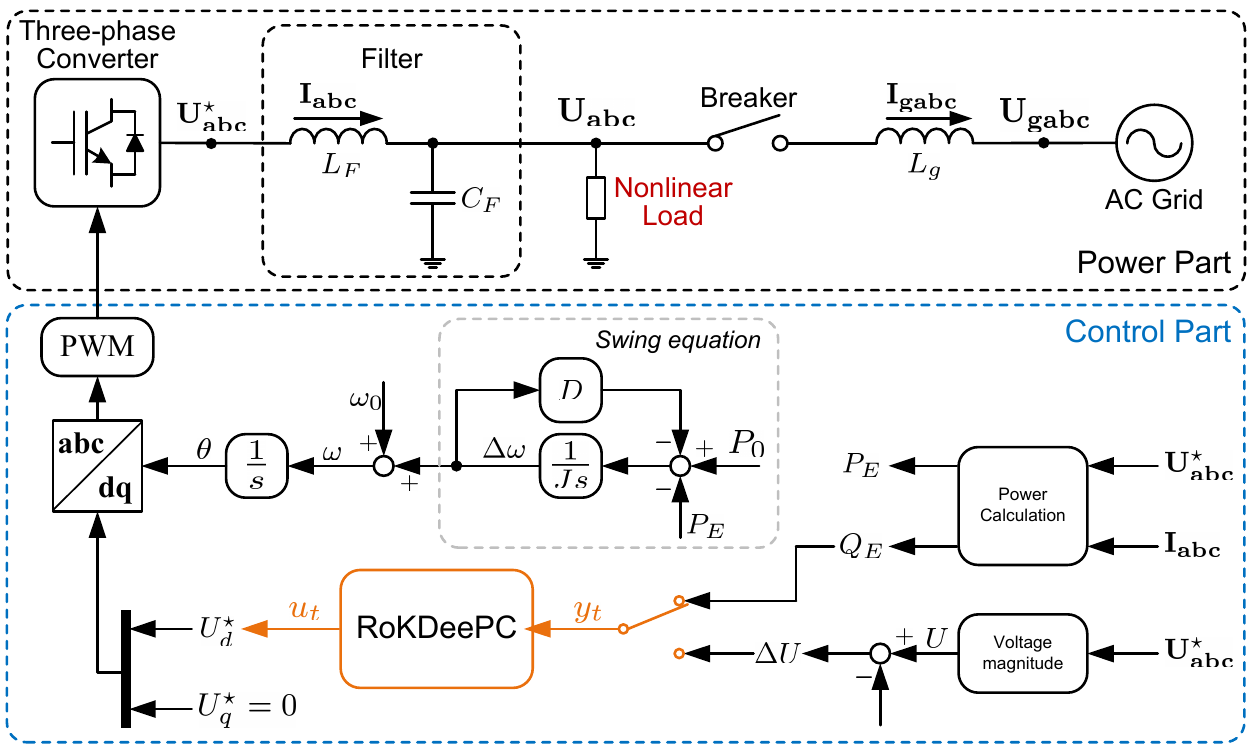}
	\vspace{-4mm}
	%\DeclareGraphicsExtensions.
	\caption{Application of RoKDeePC in a grid-forming converter feeding a nonlinear load.}
	\vspace{0mm}
	\label{Fig_converter}
\end{figure}

We apply the RoKDeePC algorithm to perform optimal voltage (or reactive power) control in the converter. As shown in Fig.~\ref{Fig_converter}, the input of the converter system is the internal voltage magnitude of the converter (i.e., $U_d^\star$), and the output can either be the reactive power $Q_E$ or the terminal voltage deviation of the load $\Delta U$. Since the converter system is more complex than~\eqref{eq:testsys1}, we choose $T_{\rm ini} = 4$, $T = 1500$, $N = 6$, and a control horizon of $k=1$. The other settings of the RoKDeePC are the same as those in the previous subsection (cost function, excitation signals, etc.). We again assume no input/output constraints and focus on comparing the tracking performance. For the Koopman-based MPC, we consider $(u_{\rm ini},y_{\rm ini})$ as the state variables of the system, and the choices of the lifting functions are the same as those in the previous subsection.

\begin{figure}[!t]
	\centering
	\includegraphics[width=3.2in]{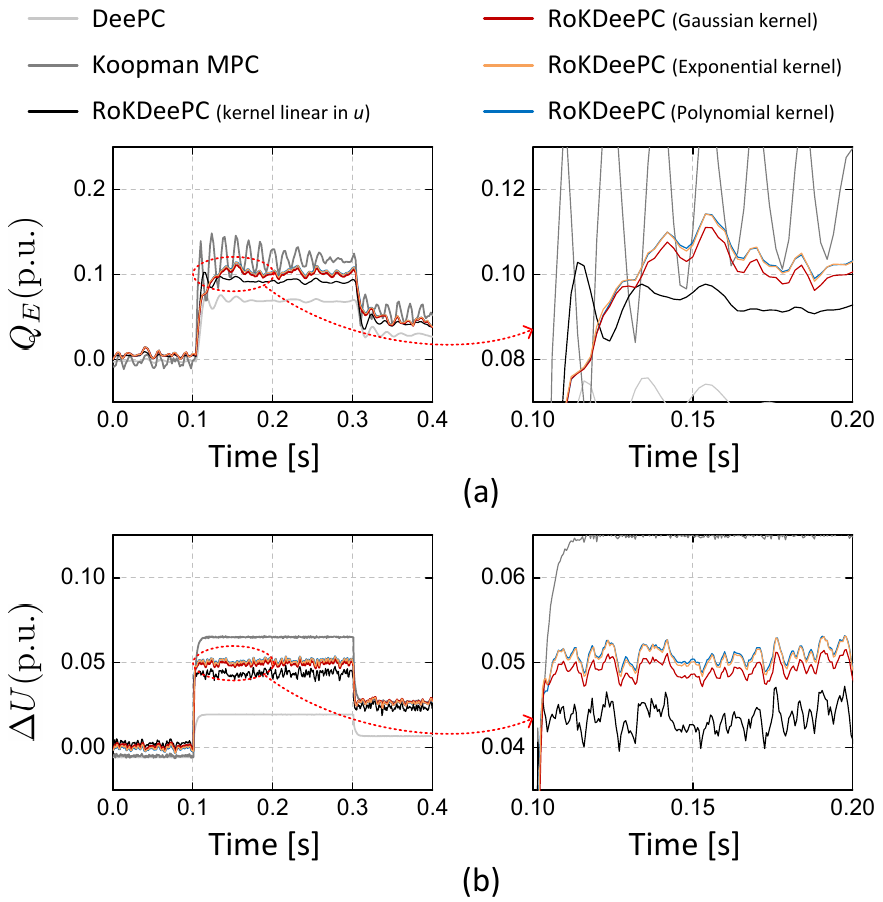}
	\vspace{-4mm}
	%\DeclareGraphicsExtensions.
	\caption{Time-domain responses of the converter system (variance of the measurement noise: $5 \times 10^{-5}$). (a) Grid-connected mode (reactive power is the output). (b) Islanded mode (voltage deviation is the output).}
	\vspace{0mm}
	\label{Fig_QV_control}
\end{figure}

Fig.~\ref{Fig_QV_control}~(a) shows the reactive power responses of the converter in grid-connected mode (i.e., the breaker in Fig.~\ref{Fig_converter} is closed), and we choose the reactive power $Q_E$ to be the output of the system. The sampling time of the system is 2ms. The reference for $Q_E$ steps from $0$ to $0.1{\rm (p.u.)}$ at $t = 0.1{\rm s}$, and then to $0.05{\rm (p.u.)}$ at $t = 0.3{\rm s}$. It can be seen that the converter has satisfactory performance when the RoKDeePC is applied (with Gaussian, Exponential, or Polynomial kernel). However, when the Koopman-based MPC or DeePC is applied, there are severe tracking errors; with the former we also observe some oscillations that persist even with more data to solve the EDMD problem and a longer prediction horizon (e.g., $T = 20000$ and $N = 20$). Similar to the results in Section~\ref{sec:VA}, this is because of the strong nonlinearity of the systems and the special structure of the lifting functions in~\eqref{eq:liftingK} that are linear in $u_t$. %For the Koopman-based MPC, we consistently observe similar oscillations and tracking error even with more data to solve the EDMD problem and a longer prediction horizon (e.g., $T = 20000$ and $N = 20$).
With the kernel-based MPC in~\eqref{eq:KMPC} (with Gaussian, Exponential, or Polynomial kernel), the system became unstable and thus the time-domain responses are not displayed in Fig.~\ref{Fig_QV_control}.
%Another possible reason is that we did not use a large amount of data to learn the Koopman-based model (only a trajectory of length 1500).

Although the RoKDeePC (with Gaussian, Exponential, or Polynomial kernel) has superior performance, it requires solving a nonlinear program in real time. In this instance, it took approximately 0.6s to solve~\eqref{eq:K_DeePC3}. By comparison, the Koopman-based MPC and DeePC require about 2ms and 80ms respectively using OSQP as the solver.

To enable a real-time implementation of the RoKDeePC, we further consider a kernel that is nonlinear in $(u_{\rm ini},y_{\rm ini})$ but linear in $u$, as shown in~\eqref{eq:kernel_linear_u}, where $\sigma^2 = 0.2$. With this kernel, \eqref{eq:K_DeePC3} becomes a quadratic program, which admits a closed-form solution when the input/output constraints are inactive.
Fig.~\ref{Fig_QV_control}~(a) also plots the reactive power response of the converter when this kernel is used, which shows better performance than the Koopman-based MPC and DeePC (with $k=1$, the tracking error is less then $10\%$). Moroever, it only requires about 10ms to calculate the closed-form solutions at each time step, and thus enables a possible real-time implementation with a control horizon $k \geq 6$.

We next test the control performance when the converter is in an islanded mode (i.e., the breaker in Fig.~\ref{Fig_converter} is open). We choose the voltage deviation $\Delta U$ to be the output of the system, that is, the RoKDeePC aims at regulating the voltage deviation of the load. Since the voltage dynamics are in general fast, we choose the sampling time to be 0.5ms. Fig.~\ref{Fig_QV_control}~(b) shows the time-domain responses of the voltage deviation when its reference steps from $0$ to $0.05{\rm (p.u.)}$ at $t = 0.1{\rm s}$, and then to $0.025{\rm (p.u.)}$ at $t = 0.3{\rm s}$. We observe that the RoKDeePC (with Gaussian, Exponential, or Polynomial kernel) achieves the best performance, whereas the Koopman-based MPC and DeePC lead to substantial tracking error. The tracking error is about $10\%$ when the kernel in~\eqref{eq:kernel_linear_u} is used in the RoKDeePC to enable faster calculations.

%Overall, we consistently observe superior performance when applying the RoKDeePC, which we attribute to its robustness against uncertainties. Moreover, the running time is generally reasonable. Note that some further modifications can be made to the gradient decent algorithm to improve the speed and optimality (e.g., using variable step size, momentum, etc.).

\section{Conclusions}

RoKDeePC, an algorithm to perform data-driven, robust, and optimal control for general nonlinear systems based on implicitly predicting the future behaviors of the system using the representer theorem was presented. RoKDeePC involves solving a data-to-control optimization problem in a receding-horizon manner, which does not require any time-consuming off-line learning step. Moreover, we demonstrate how to provide end-to-end robustification for the control sequence against measurement noise in the output data, leading to strong performance guarantees. To handle the nonconvexity of the optimization problem, we exploited its structure and developed projected gradient descent algorithms to enable fast calculations of the (sub)optimal solutions. The RoKDeePC was tested in two example nonlinear systems (including a high-fidelity converter model feeding a nonlinear load) and showed excellent performance with reasonable running time, compatible with real-time implementations in real-world applications. A comparison with related nonlinear data-driven MPC methods showed the favorable performance of our approach.

\section*{Acknowledgment}

The authors would like to thank Liviu Aolaritei, Francesco Micheli, and Jianzhe Zhen for fruitful discussions.

\appendices
\section{Proof of Proposition~\ref{Prop1}}
\label{proof_prop1}

\begin{proof}
%The proof is adapted from~\cite[Theorem~1]{huang2020quad}. Here we assume that $\mathcal{G}$ is a polyhedron, i.e., $\mathcal{G} = \{ g \; | \; Gg \le q \}$ to simplify the proof.

Since $\mathcal{Y}$ is a polyhedron, we can compactly rewrite $\mathcal{G}$ as $\mathcal{G} = \{g \ | \ Gg+ \mathcal{N}_G(g) \le q\}$ where $\mathcal{N}_G(g) = [\|G^{(1)}g\|\; \|G^{(2)}g\|\; \cdots \; \|G^{(n_q)}g\|]^\top$ for some $G \in \mathbb{R}^{n_q \times H_c}$, $q \in \mathbb{R}^{n_q}$, and $G^{(i)} \in \mathbb{R}^{pNH_c \times H_c},\; \forall i \in [n_q]$~\cite{bdv15}.

We start by respectively rewriting~\eqref{eq:min_cq} and~\eqref{eq:min_c} as
\begin{equation}\label{eq:cq1}
\mathop{\min} \limits_{g \in \mathcal{G}}\; \|\hat Y_{\rm F}g - r\|_Q^2 + \lambda_g \|g\|^2
+ \lambda_k'\| Vg-v \|^2 + c_u \,,
\end{equation}
\begin{equation}\label{eq:c1}
\mathop{\min} \limits_{g \in \mathcal{G}}\; \|\hat Y_{\rm F}g - r\|_Q + h(g)
+ \lambda_k\| Vg-v \| + c_u \,,
\end{equation}
where $c_u = \ell(u)$ is a constant, $V=\mathbf{\hat K}+\gamma I$, $v=k(u_{\rm ini},\hat y_{\rm ini},u)$, and $h(g) = \lambda_k\rho_1 \sqrt{\|g\|^2+1} + \rho_2 \|g\|$.

Consider the Lagrangian of~\eqref{eq:cq1}
\begin{equation*}
\begin{split}
\mathcal{L}_q(g,\mu_q) = & \hspace{0.8mm} \|\hat Y_{\rm F}g - r\|_Q^2 + \lambda_g \|g\|^2
+ \lambda_k'\| Vg-v \|^2 \\
& + \mu_q^\top (Gg+\mathcal{N}_G(g)-q) \,,
\end{split}
\end{equation*}
where $\mu_q$ is the vector of the dual variables. Since $\mathcal{G}$ is nonempty, there exists a solution $(g^\star,\mu_q^\star)$ to the Karush–Kuhn–Tucker (KKT) conditions of~\eqref{eq:cq1}
\begin{subequations}
    \begin{empheq}[left={\empheqlbrace}]{align}
	& \begin{array}{r} \hspace{-1.6mm}
2 \hat Y_{\rm F}^\top Q (\hat Y_{\rm F}g - r) + 2 \lambda_g g  + 2\lambda'_k V^\top (Vg-v)  \\
+ \hspace{0.5mm} G^\top \mu_q + \sum_{i \in [n_q]}\mu_{q[i]}  \frac{(G^{(i)})^\top G^{(i)}g}{\|G^{(i)}g\|} = 0 \,,
\end{array}
\label{eq:KKT_LQ1}\\
	& \mu_q^\top (Gg + \mathcal{N}_G(g) - q) = 0  \,,   \label{eq:KKT_LQ2}  \\
	& Gg + \mathcal{N}_G(g) \le  q            \,,     \label{eq:KKT_LQ3}  \\
	& \mu_q \ge 0          \,.       \label{eq:KKT_LQ4}
    \end{empheq}
    \end{subequations}
where $g^\star$ is a minimizer of~\eqref{eq:cq1} (and~\eqref{eq:min_cq}).

Consider the Lagrangian of~\eqref{eq:c1}
\begin{equation*}
\begin{split}
\mathcal{L}(g,\mu) = & \hspace{0.8mm} \|\hat Y_{\rm F}g - r\|_Q + \lambda_k\rho_1 \sqrt{\|g\|^2+1} + \rho_2 \|g\|
\\
& + \lambda_k\| Vg-v \| + \mu^\top (Gg+\mathcal{N}_G(g)-q) \,,
\end{split}
\end{equation*}
where $\mu$ is the vector of the dual variables. By choosing $\lambda_k$, $\rho_1$, and $\rho_2$ according to~\eqref{eq:lambda_k_eq} and~\eqref{eq:lambda_g_eq}, it can be verified that $(g^\star, \mu^\star, y^\star, w^\star, z^\star)$, where
\begin{equation*}
	(\mu^\star,  y^\star)  = \begin{cases}
	
	\left( \frac{\mu_q^\star}{2 \|\hat Y_{\rm F}g^\star - r\|_Q}, \frac{\hat Y_{\rm F}^\top Q(\hat Y_{\rm F}g^\star - r)}{\|\hat Y_{\rm F}g^\star - r\|_Q} \right)  \quad
	\text{if $\hat Y_{\rm F}g^\star \ne r$,}\\
	\left(\frac{\mu_q^\star}{2} , 0 \right) \quad  \text{otherwise,} \end{cases}
\end{equation*}
\begin{equation*}
	z^\star  = \begin{cases}
	\frac{\lambda_k V^\top (Vg^\star-v)}{\|Vg^\star-v\|}   \quad
	\text{if $Vg^\star \ne v$,}\\
	0 \quad  \text{otherwise,} \end{cases}
\end{equation*}
$w^\star = \rho_2 g^\star / \|g^\star\|$ if $g^\star \ne 0$, and $w^\star = 0$ otherwise,
and $g^\star$ as before, satisfy the KKT conditions of~\eqref{eq:c1}
\begin{subequations}
    \begin{empheq}[left={\empheqlbrace}]{align}
	& \hspace{-1.5mm} \begin{array}{l} y + \frac{\lambda_k \rho_1 g}{\sqrt{\|g\|^2+1}} + w + z + G^\top \mu \\
    \hspace{5mm} + \sum_{i \in [n_q]}\mu_{[i]}  \frac{(G^{(i)})^\top G^{(i)}g}{\|G^{(i)}g\|} = 0 \,,
    \end{array}
\label{eq:KKT_L1}\\
    & \hspace{-2mm} \begin{array}{r} \|\hat Y_{\rm F}\bar g - r\|_Q \ge \|\hat Y_{\rm F}g - r\|_Q + y^\top (\bar g - g), \\ \forall \bar g \in \mathbb{R}^{H_c},
    \end{array}
\label{eq:KKT_L2}  \\ 
	& \rho_2 \| \bar g\| \ge \rho_2 \|g\| + w^\top (\bar g - g), \; \forall \bar g \in \mathbb{R}^{H_c} ,
\label{eq:KKT_L3}  \\
	& \hspace{-2mm} \begin{array}{r} \lambda_k \|V g - v\| \ge \lambda_k \|V g - v\| + z^\top (\bar g - g), \\ \forall \bar g \in \mathbb{R}^{H_c},
    \end{array}
\label{eq:KKT_L4}  \\
    & \mu^\top (Gg+ \mathcal{N}_G(g)-q) = 0 \,,
\label{eq:KKT_L5}  \\
    & Gg+ \mathcal{N}_G(g) \le q \,,
\label{eq:KKT_L6}  \\
    & \mu \ge 0 \,.
\label{eq:KKT_L7}
    \end{empheq}
    \end{subequations}
Thus, the vector $g^\star$ is also a minimizer of~\eqref{eq:c1} (and~\eqref{eq:min_c}).

Next we prove the monotonic relationship  between $\lambda_k$ and $\lambda_k'$. Let $g_1$ be the minimizer of~\eqref{eq:cq1} with $\lambda_k' = \lambda_{k1}' >0$ (and the minimizer of~\eqref{eq:c1} with $\lambda_k = \lambda_{k1}$), and $g_2$ the minimizer of~\eqref{eq:cq1} with $\lambda_k' = \lambda_{k2}' > \lambda_{k1}'$ (and the minimizer of~\eqref{eq:c1} with $\lambda_k = \lambda_{k2}$). If $g_1 = g_2$, we directly obtain $\lambda_{k1} < \lambda_{k2}$ according to~\eqref{eq:lambda_k_eq}.
If $g_1 \ne g_2$, according to the definitions of $g_1$ and $g_2$, we have
\begin{equation*}
\begin{split}
& \|\hat Y_{\rm F}g_1 - r\|_Q^2 + \lambda_g \|g_1\|^2
+ \lambda_{k1}'\| Vg_1-v \|^2 \\
<\; & \|\hat Y_{\rm F}g_2 - r\|_Q^2 + \lambda_g \|g_2\|^2
+ \lambda_{k1}'\| Vg_2-v \|^2 \,,
\end{split}
\end{equation*}
which can be reorganized as
\begin{equation}\label{eq:g1g2_1}
\begin{split}
& \|\hat Y_{\rm F}g_1 - r\|_Q^2 + \lambda_g \|g_1\|^2
+ \lambda_{k1}'(\| Vg_1-v \|^2 - \| Vg_2-v \|^2) \\
& <  \|\hat Y_{\rm F}g_2 - r\|_Q^2 + \lambda_g \|g_2\|^2 \,.
\end{split}
\end{equation}
Moreover, we have
\begin{equation*}
\begin{split}
& \|\hat Y_{\rm F}g_2 - r\|_Q^2 + \lambda_g \|g_2\|^2
+ \lambda_{k2}'\| Vg_2-v \|^2 \\
<\; & \|\hat Y_{\rm F}g_1 - r\|_Q^2 + \lambda_g \|g_1\|^2
+ \lambda_{k2}'\| Vg_1-v \|^2 \,,
\end{split}
\end{equation*}
which can be reorganized as
\begin{equation}\label{eq:g1g2_2}
\begin{split}
& \|\hat Y_{\rm F}g_1 - r\|_Q^2 + \lambda_g \|g_1\|^2
+ \lambda_{k2}'(\| Vg_1-v \|^2 - \| Vg_2-v \|^2) \\
&> \|\hat Y_{\rm F}g_2 - r\|_Q^2 + \lambda_g \|g_2\|^2 \,.
\end{split}
\end{equation}
By combining~\eqref{eq:g1g2_1} and~\eqref{eq:g1g2_2}, we obtain
\begin{equation*}
\begin{split}
& \lambda_{k1}'(\| Vg_1-v \|^2 - \| Vg_2-v \|^2) \\
< \; & \lambda_{k2}'(\| Vg_1-v \|^2 - \| Vg_2-v \|^2) \,,
\end{split}
\end{equation*}
which indicates that $\| Vg_1-v \|^2 > \| Vg_2-v \|^2$ as $\lambda_{k2}' > \lambda_{k1}' > 0$. Then, we have $\| Vg_1-v \| > \| Vg_2-v \|$.

According to the definitions of $\lambda_{k1}$ and $\lambda_{k2}$, we also have
\begin{equation*}
\begin{split}
& \|\hat Y_{\rm F}g_1 - r\|_Q + h(g_1)
+ \lambda_{k1}(\| Vg_1-v \| - \| Vg_2-v \|) \\
& <  \|\hat Y_{\rm F}g_2 - r\|_Q + h(g_2) \,,
\end{split}
\end{equation*}
\begin{equation*}
\begin{split}
& \|\hat Y_{\rm F}g_1 - r\|_Q + h(g_1)
+ \lambda_{k2}(\| Vg_1-v \| - \| Vg_2-v \|) \\
& >  \|\hat Y_{\rm F}g_2 - r\|_Q + h(g_2) \,,
\end{split}
\end{equation*}
leading to
\begin{equation*}
\begin{split}
& \lambda_{k1} (\underbrace{\| Vg_1-v \| - \| Vg_2-v \|}_{>0})  \\
< \hspace{0.6mm} & \lambda_{k2}(\underbrace{\| Vg_1-v \| - \| Vg_2-v \|}_{>0})
\end{split}
\end{equation*}
It can then be deduced that $\lambda_{k1} < \lambda_{k2}$. Hence, $\lambda_k$ is increasing with the increase of $\lambda_k'$. By following a similar process, one can also prove that $\rho_1$ ($\rho_2$) is increasing with the increase of $\lambda_g$. This completes the proof.
\end{proof}

	\normalem
    \bibliographystyle{IEEEtran}
	\bibliography{references}

\end{document}